\documentclass[aps,twocolumn,superscriptaddress,nofootinbib,a4paper,longbibliography]{revtex4-2}
\usepackage{times}
\usepackage{epsfig}
\usepackage{amsfonts}
\usepackage{amsmath}
\usepackage{amsthm}
\usepackage{amssymb}					
\usepackage{amsthm}
\usepackage{dsfont}
\usepackage{bm}
\usepackage{enumerate}
\usepackage{mathtools}
\usepackage{color}
\usepackage{multirow}
\usepackage[normalem]{ulem}
\newcommand{\stkout}[1]{\ifmmode\text{\sout{\ensuremath{#1}}}\else\sout{#1}\fi}
\usepackage{latexsym}
\usepackage{mathrsfs}
\usepackage{natbib}
\usepackage{verbatim}
\usepackage[T1]{fontenc}
\usepackage{float}
\usepackage{graphicx}
\usepackage{subcaption}
\usepackage{xcolor}
\usepackage{physics}
\usepackage{soul}
\usepackage{caption}
\usepackage{subcaption}

\usepackage[font=small,labelfont=bf]{caption}

\usepackage{xspace}

\newtheorem{theorem}{Theorem}

\newcommand{\bracket}[3]{\langle#1|#2|#3\rangle}
\newcommand{\expect}[1]{\langle#1\rangle}

\usepackage{amssymb}
\usepackage{amsthm}
\usepackage{mathtools} %do robienia mini macierzy
\usepackage[customcolors]{hf-tikz} % do kolorowych macierzy
\usetikzlibrary{patterns}
\usetikzlibrary{matrix,decorations.pathreplacing}

\pgfkeys{tikz/mymatrixenv/.style={decoration={brace},every left delimiter/.style={xshift=8pt},every right delimiter/.style={xshift=-8pt}}}
\pgfkeys{tikz/mymatrix/.style={matrix of math nodes,nodes in empty cells,left delimiter={[},right delimiter={]},inner sep=1pt,outer sep=1.5pt,column sep=2pt,row sep=2pt,nodes={minimum width=20pt,minimum height=10pt,anchor=center,inner sep=0pt,outer sep=0pt}}}
\pgfkeys{tikz/mymatrixbrace/.style={decorate,thick}}

\usepackage[colorlinks=true,linkcolor=blue,citecolor=magenta,urlcolor=blue]{hyperref}

\begin{document}
	
	%%%%%%%%%%%%%%%%%%%%%%%%%%%%%%%%%%%%%%%%%%%%%%%%%%%%%%%%%%%%%%%%%%%

\title{Simulating quantum instruments with projective measurements and quantum post-processing}

\author{Shishir Khandelwal}
\affiliation{Physics Department and NanoLund, Lund University, Box 118, 22100 Lund, Sweden.}

\author{Armin Tavakoli}
\affiliation{Physics Department and NanoLund, Lund University, Box 118, 22100 Lund, Sweden.}

\begin{abstract}
Quantum instruments  describe both the classical outcome and the updated state associated with a quantum measurement.  We ask whether these processes can be simulated using only a natural subset of resources, namely projective measurements on the system and quantum processing of the post-measurement states. We show that the simulability of instruments can be connected to an entanglement classification problem. This leads to a computationally efficient necessary condition for simulation of generic instruments and to a complete characterisation for qubits. We use this to address relevant quantum information tasks, namely (i) the noise-tolerance of standard qubit unsharp measurements, (ii) non-projective advantages in  information-disturbance trade-offs, and (iii) increased sequential Bell inequality violations under projective measurements. Moreover, we consider also $d$-dimensional L\"uders instruments that correspond to weak versions of standard basis measurements and show that for large $d$ these can permit scalable noise-advantages over projective implementations. 
\end{abstract}

\date{\today}

\maketitle

\section{Introduction}

It is often relevant to ask if complex quantum resources can be simulated by simpler quantum resources \cite{Chitambar2019}. An important example of this concerns quantum measurements. Measurements are commonly associated with orthogonal projections, but their  full scope corresponds to the more general notion of positive operator-valued measures (POVMs). Non-projective measurements play an important role in quantum information protocols, but their physical implementation is more complex than that of projective measurements. It is therefore relevant to ask which POVMs  can be simulated using projective measurements and which cannot \cite{Oszmaniec2017}. In recent years, many experiments have realised POVMs that defy simulation with projective measurements \cite{Gomez2016, Tavakoli2020, Smania2020, Martinez2023, Wang2023, Feng2023}.

While POVMs describe the classical output of a measurement, they do not describe the update of the quantum state. A complete description of the measurement process must address both the classical and quantum output. This is known as a quantum instrument \cite{QMeasurement} (see Fig.~\ref{Fig_scenario}a). Quantum instruments are not only an essential component of quantum theory, but they are also responsible for the fundamental trade-off between extracting information from a system and disturbing its state. Experiments have  observed these trade-offs and demonstrated their applications \cite{Andersen2006, Lim2014, Schiavon2017, Anwer2020, Foletto2020, Anwer2021}. 
\begin{figure}[t!]
	\centering
	\includegraphics[width=0.9\columnwidth]{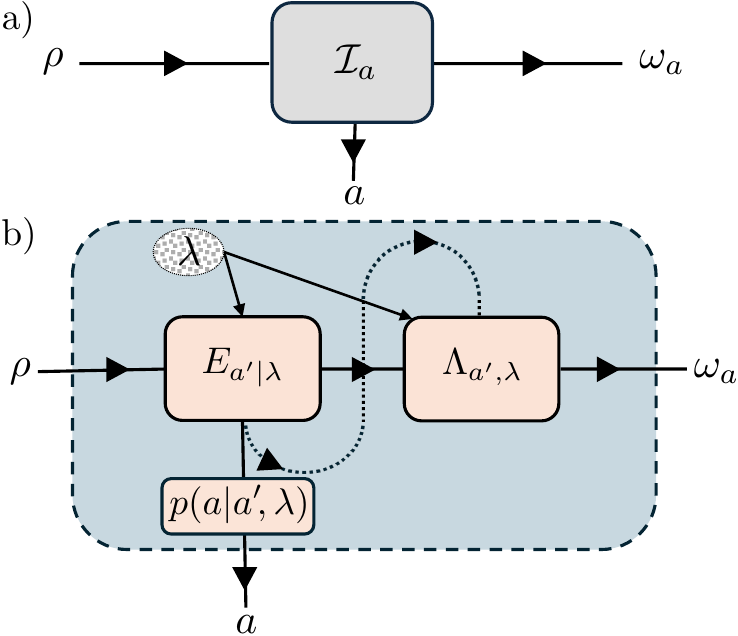}
	\caption{\textbf{Quantum instruments.}  (a) A quantum instrument transforms an incoming state $\rho$ into a classical outcome $a$ and an updated quantum state $\omega_a$. (b) A projective simulable instrument uses a classical random variable, $\lambda$, to select a projective measurement, $\{E_{a'|\lambda}\}$, then passes the projected state through a quantum channel, $\Lambda_{a',\lambda}$, and finally uses a rule $p(a|a',\lambda)$ to decide the classical output.}\label{Fig_scenario}
\end{figure}

Here, we ask whether quantum instruments can be simulated by an experimenter who only has access to standard projective measurements on the system, classical post-processing of the measurement outcomes, quantum processing of the post-measurement state and classical randomness; see Fig.~\ref{Fig_scenario}b. Thus, these simulations do not have access to the coherent non-projective capabilities associated with generic instruments.  This can be seen as a natural generalisation from projective simulation of POVMs to now account for the full measurement process. We show that the characterisation of simulable instruments can be relaxed to a type of entanglement classification problem. Using semidefinite programming (SDP) techniques (see the review \cite{Tavakoli2024}), this allows us to efficiently determine proofs of genuinely non-projective instruments. Furthermore, for qubit systems, we prove that this characterisation is both necessary and sufficient. This makes it a handy tool, which we showcase by considering three quantum information problems. Firstly, we consider noisy  qubit instruments, corresponding to unsharp  measurements, and determine the visibilities needed to make a projective simulation possible. Secondly, we show how non-projectivity implies advantages in information-disturbance trade-offs. Thirdly, we show that sequential Bell inequality tests, that use only projective measurements and local randomness \cite{Steffinlongo2022}, can admit significantly larger violations than previously known.  Moreover, we also address the role of Hilbert space dimension. Specifically, we consider instruments that weakly implement high-dimensional basis measurements and ask whether the large dimensionality implies scalable noise-advantages over projectively simulable instruments. We find that for natural noise forms, scalable advantages are possible, but also that there exist noise forms for which the same is not true.

\section{Preliminaries}
We denote Hilbert spaces by $\mathcal{H}$, the space of density matrices as $\mathcal{D}(\mathcal{H})$ and the space of positive semidefinite linear operators as $\mathcal{L}_+(\mathcal{H})$. A quantum instrument, $\boldsymbol{\mathcal{I}}$, transforms a  state, $\rho\in\mathcal{D}(\mathcal{H}_A)$, into a classical outcome, $a$, and a corresponding post-measurement state, $\omega_a\in\mathcal{D}(\mathcal{H}_{A'})$. It is represented as a set $\boldsymbol{\mathcal{I}}=\{\mathcal{I}_a\}_{a=1}^N$, where each $\mathcal{I}_a: \mathcal{H}_A \rightarrow \mathcal{H}_{A'}$ is a completely positive and trace non-increasing map. The  post-measurement state is given by $\omega_a=\frac{\mathcal{I}_a(\rho)}{\tr\left(\mathcal{I}_a(\rho)\right)}$ where the normalisation is the probability of outcome $a$, namely $p(a)=\tr(\mathcal{I}_a(\rho))$. Normalisation of the probability distribution means that the map $\sum_a \mathcal{I}_a$ is trace-preserving. Using state-channel duality \cite{QMeasurement}, instruments can be represented using Choi operators. Each $\mathcal{I}_a$ is then associated with a bipartite operator $\eta_a\in \mathcal{L}_+(\mathcal{H}_{A'}\otimes \mathcal{H}_A)$ defined as $\eta_a=(\mathcal{I}_a\otimes \openone)[\phi^+]$, where  $\phi^+=\ketbra{\phi^+}$ and $\ket{\phi^+}=\frac{1}{\sqrt{d}}\sum_{i=0}^{d-1}\ket{ii}$ is the maximally entangled state of dimension $d=\text{dim}(\mathcal{H}_A)$. Complete positivity and normalisation are equivalent to $\eta_a\succeq 0$ and $\tr_{A'}\sum_a \eta_a=\frac{\mathds{1}}{d}$ respectively. The (sub-normalised) quantum output of the instrument can be written as $\mathcal{I}_a(\rho)=d \tr_{A}\left((\openone_{A'}\otimes \rho_A^T)\eta_a\right)$. The POVM, $\{M_a\}_a$, that is realised by the instrument  is obtained from its reduced Choi operator as $M_a=d\tr_{A'}(\eta_a)^T$.

Consider now an experimenter that has access only to the following three resources: (i) classical randomness and classical post-processing, (ii) projective measurements on $\mathcal{H}_A$ and (iii) quantum post-processing. Thus, the experimenter  draws a classical variable $\lambda$ from some distribution $\{q_\lambda\}_\lambda$, selects a projective measurement $\{E_{a'|\lambda}\}$ with outcome $a'$ ($E_{a'|\lambda}^2=E_{a'|\lambda}$ and $\sum_{a'} E_{a'|\lambda}=\openone$), and then implements a quantum channel $\Lambda_{a',\lambda}:\mathcal{H}_A\rightarrow \mathcal{H}_{A'}$ on the (subnormalised) post-projection state $E_{a'|\lambda}\rho E_{a'|\lambda}$. A post-processing rule $p(a|a',\lambda)$ is used to determine the final classical output. This is illustrated in Fig~\ref{Fig_scenario}b. Instrument of this type are written as
\begin{equation}\label{PI}
	\mathcal{I}_a(\rho)=\sum_\lambda q_\lambda\ \sum_{a'} p(a|a',\lambda)\ \Lambda_{a',\lambda}\big[E_{a'|\lambda}\rho E_{a'|\lambda}\big],
\end{equation}
We refer to these as \textit{projective instruments} (PIs) and denote their set as $\mathcal{P}$. In Appendix \ref{app:post}, we prove that the post-processing can be discarded w.l.g., i.e.~that one may restrict to $p(a|a',\lambda)=\delta_{a,a'}$, and that one can w.l.g.~also restrict  to $\Lambda_{a',\lambda}=\Lambda_{\lambda}$.

\section{Simulability of instruments} 
How to determine whether an instrument is PI-simulable?  We begin with giving a general necessary condition.  To this end, let us associate every $N$-outcome projective measurement, $\{E_a\}_{a=1}^N$, with a rank-vector $\vec{r}=(r_1,\ldots,r_N)$, where $r_a=\rank(E_a)$. Thus, every $N$-tuple of non-negative integers such that $\sum_a r_a=d$ is a valid rank-vector. We can separately consider the measurements associated with each such $\vec{r}$. Therefore, we write $\lambda=(\chi,\vec{r})$, where $\chi$ is a random variable for selecting projective measurements with rank-vector $\vec{r}$. Consider now the Choi representation of the generic PI defined in Eq.~\eqref{PI} with $p(a|a',\lambda)=\delta_{a,a'}$ and $\Lambda_{a',\lambda}=\Lambda_{\lambda}$. It reads 
\begin{equation}\label{PIchoi}
	\eta_a\!=\sum_{\chi,\vec{r}} q_{\chi,\vec{r}} \left(\Lambda_{\chi,\vec{r}}\otimes\openone\right)\left[\nu_{a|\chi,\vec{r}}\right],
\end{equation}
where $\nu_{a|\chi,\vec{r}}=(E_{a|\chi,\vec{r}}\otimes \openone)\phi^+(E_{a|\chi,\vec{r}}\otimes \openone)$. Observe that (i)  the local projection $E_{a|\chi,\vec{r}}$ of $\phi^+$ succeeds with probability $\frac{r_a}{d}$, and (ii) the operator $\nu_{a|\chi,\vec{r}}$ is locally confined to an $r_a$-dimensional subspace. Consequently, the entanglement dimension (a.k.a~its Schmidt number, SN\footnote{The Schmidt number of a bipartite state $\varphi$ is the smallest integer, $s$, such that $\varphi=\sum_{i} p_i \ketbra{\psi_i}$ with $\rank(\tr_{A}(\ketbra{\psi_i}))\leq s\, \forall i$.} \cite{Terhal2000}) of $\nu_{a|\chi,\vec{r}}$ is at most $r_a$.  This provides the intuition for our first result (see Appendix \ref{app:thm1} for details).

\begin{theorem}[Instrument simulation]\label{thm1}
	Consider any $N$-outcome projective instrument and let  $\{\eta_a\}_{a=1}^N$ be its  Choi representation. There exists a decomposition 
	\begin{align}\label{choi}\nonumber
		&\qquad  \eta_a= \sum_{\vec{r}} \sigma_{a|\vec{r}}, \quad \text{where} \quad \sigma_{a|\vec{r}}\in\mathcal{L}_+(\mathcal{H}_{A'}\otimes\mathcal{H}_A)\\
		&\tr(\sigma_{a|\vec{r}})=q_{\vec{r}}\frac{r_a}{d},\quad  \sum_{a}\tr_{A'}(\sigma_{a|\vec{r}})=q_{\vec{r}}\frac{\openone}{d}, \quad  \text{SN}(\sigma_{a|\vec{r}})\leq r_a.
	\end{align}
	where $\vec{r}$ runs over all rank-vectors and $\text{SN}$ denotes the Schmidt number. 
\end{theorem}

%The Choi representation of a PMI reads .	
%First, we are locally projecting the $d$-dimensional maximally entangled state onto an $r_a$-dimensional subspace. This projection succeeds with probability $\frac{r_a}{d}$, and the post-projection state is effectively of dimensions $r_a\times d$. Trivially, such a state cannot have an entanglement dimension (Schmidt number) greater than $r_a$. This is seen directly from the definition of the Schmidt number \cite{Terhal2000}:   Since the Schmidt number is an entanglement measure, the local channel $\Lambda_{a,\chi,\vec{r}}$ cannot increase it. Therefore, $\sigma_{a|\vec{r}}\equiv \frac{d}{r_a}\sum_\chi q_{\chi,\vec{r}} (\Lambda_{a,\chi,\vec{r}}\otimes\openone)(E_{a|\chi,\vec{r}}\otimes \openone)\phi^+(E_{a|\chi,\vec{r}}\otimes \openone)$ is a state with $\text{SN}(\sigma_{a|\vec{r}})\leq r_a$. The Choi representation takes the form \eqref{choi} where $q_{\vec{r}}=\sum_{\chi}q_{\chi,\vec{r}}$. The second criterion in \eqref{choi} follows from normalisation. 

We can interpret the operators $\{\sigma_{a|\vec{r}}\}_a$ as corresponding to a sub-normalised Choi representation of a PI that uses only measurements with rank-vector $\vec{r}$. The sub-normalisation, $q_{\vec{r}}$, corresponds to the probability of selecting $\vec{r}$ and it is given by $q_{\vec{r}}=\sum_a \tr(\sigma_{a|\vec{r}})$. Thus, Theorem~\ref{thm1} states that if the instrument $\boldsymbol{\mathcal{I}}$ is not a convex combination over the Choi representations associated with the different rank-vectors, then $\boldsymbol{\mathcal{I}}\notin\mathcal{P}$ and hence it is genuinely non-projective.

Theorem~\ref{thm1} is not a sufficient  condition for simulability because not all Choi representations with a Schmidt number as in Eq.~\eqref{choi} can be associated with PIs. Importantly, however, for the practically most relevant case, namely qubits ($d=2$), it turns out to be an exact characterisation.

\begin{theorem}[Qubit instrument simulation]\label{thm2}
	For instruments with qubit input, Theorem~\ref{thm1} is both necessary and sufficient.
\end{theorem}
\begin{proof}
	We give the main idea here and present details in Appendix \ref{app:thm2}. For qubits, there are only two rank-vectors; $\vec{s}=(2,0,\ldots,0)$ and $\vec{t}=(1,1,0,\ldots,0)$, up to permutations. Thus, $\nu_{a|\chi,\vec{s}}=\phi^+$ for $a=1$, yielding arbitrary Choi operators in \eqref{PIchoi}, and $\eta_a=0$ for  $a\neq 1$. This corresponds to \eqref{choi}, since $\text{SN}\leq 2$ trivially holds for any density matrix over $\mathcal{D}(\mathcal{H}_{A'}\otimes \mathbb{C}^2)$. $\vec{t}$ leads to $\nu_{a|\chi,\vec{t}}=\frac{1}{2}E_{a|\chi,\vec{t}}\otimes E_{a|\chi,\vec{t}}^T$ for $a\in\{1,2\}$ and via \eqref{PIchoi} to Choi operators $\frac{1}{2}\sum_{\chi}q_{\chi,\vec{t}} \varphi_{a,\chi,\vec{t}}\otimes  E^T_{a|\chi,\vec{t}}$, for some arbitrary  $\varphi_{a,\chi,\vec{t}}$. Any sub-normalised separable Choi operator can be written on this form and it is  equivalent to the characterisation in \eqref{choi} for $\vec{t}$. 
\end{proof}

% For these, the only non-zero Choi operator is $\sigma_{1|(2,0)}$ and $\sigma_{2|(0,2)}$ repsectively, and the Schmidt number constraint in Eq.~\eqref{cond} is trivially satisfied since $\sigma_{a|\vec{r}}\in\mathcal{L}_+(\mathbb{C}^2\otimes \mathbb{C}^2)$. The final rank-vector corresponds to a basis measurement and $\text{SN}(\sigma_{a|(1,1)})\leq1$ is equivalent to $\sigma_{a|(1,1)}$ being separable.

%For qubits there are only two qualitative types of projective measurements; with rank-vectors corresponding to $(2,0,\ldots,0)$ and $(1,1,0,\ldots,0)$ up to permutations. In the former case, only the Choi operator $\sigma_{1|(2,0,\ldots,0)}$ is non-zero and $\text{SN}(\sigma_{1|(2,0,\ldots,0)})\leq 2$ is trivially satisfied by every $\sigma\in\mathcal{L}_+(\mathcal{H}_2\otimes\mathcal{H}_1)$.

Theorem~\ref{thm2} reduces the full characterisation of PIs acting on qubits to a type of separability problem. For qubit-qubit systems, separability is equivalent to positive partial-transpose \cite{Horodecki1996}. This means that for instruments whose input and output are qubits, we can replace in Eq.~\eqref{choi} the only non-trivial Schmidt number constraint ($\text{SN}\leq1$) with  $\sigma_{a|\vec{r}}^{T_A}\succeq 0$. This is crucial since it makes Eq.~\eqref{choi} an SDP. This allows us to efficiently decide the membership problem to $\mathcal{P}$. 

Building on Theorem~\ref{thm1}, SDPs can also be used to falsify simulability for generic instruments. Note that the only condition in \eqref{choi} that is not SDP-compatible is that concerning the Schmidt number. While there exists no computable necessary and sufficient condition for Schmidt number characterisation \cite{Gharibian2010}, many partial criteria are known (see e.g.~\cite{Terhal2000, Sperling2011, Shahandeh2014, Weilenmann2020, Morelli2023, Tavakoli2024b}). Drawing on this literature, we employ such a partial Schmidt number criterion which  admits a semidefinite formulation \cite{Cobucci2024}, as this results in an SDP relaxation of $\mathcal{P}$. A practical choice is based on the reduction map $\Theta_s(X)=\tr(X)\openone-\frac{1}{s}X$. It has the property that  $(\Theta_s\otimes\openone)[\sigma]\succeq 0$ for all $\sigma$ with $\text{SN}\leq s$ \cite{Tomiyama1985, Terhal2000}. Thus, if the SDP is infeasible, no simulation exists. Our implementation is available at \cite{OurCode}.

\begin{figure}[t!]
	\centering
	\includegraphics[width=0.9\columnwidth]{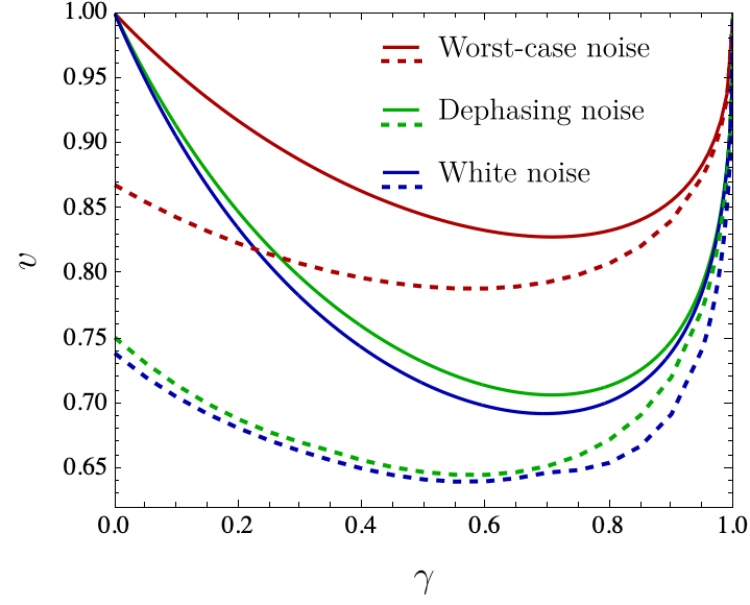}
	\caption{\textbf{Critical visibilities for simulation.} Solid lines are the critical visibilities for PI-simulation of L\"uders instruments for qubit observables with sharpness $\gamma$. Dashed lines are upper bounds on the criticial visibility for L\"uders instruments for qutrit measurements.}\label{weakZ}
\end{figure}

\section{Applications}
We now proceed to demonstrate the practical relevance of these methods via three well-known quantum information applications. These concern, respectively, unsharp measurements, information-disturbance relations and nonlocality.

\subsection{Unsharp observables}\label{secX}
Consider an instrument realising an unsharp Pauli observable with eigenbasis $\{\ket{0},\ket{1}\}$. The associated L\"uders instrument has Kraus operators $K_{a}=\sqrt{\frac{1-(-1)^a \gamma}{2}}\ketbra{0}{0}+\sqrt{\frac{1+(-1)^a \gamma}{2}}\ketbra{1}{1}$, for outcomes $a\in\{1,2\}$ and sharpness parameter $\gamma\in[0,1]$. Thus, $\mathcal{I}_a(\rho)=K_a\rho K_a^\dagger$ and its Choi representation is $\eta_a=(K_a\otimes\openone)[\phi^+](K_a^\dagger\otimes \openone)$. This non-projective qubit instrument is extremal \cite{Pellonpaa2013} and frequently used in  quantum information \cite{Fuchs1996, Silva2015, Mohan2019, Anwer2021}. The associated POVM is $M_a=K_a^\dagger K_a$ and it is always projective simulable, but the instrument is PI only when $\gamma\in\{0,1\}$.

We determine the instrument's  best  approximation by a PI. To this end, we  find the precise amount of noise that these instruments must be exposed to in order to make them PI-simulable for any given $\gamma$. Thus, we consider mixtures
\begin{equation}\label{noise}
	\eta_{a}^v\equiv v\eta_a+(1-v)\eta_a^\text{noise},
\end{equation}
where $\{\eta_a^\text{noise}\}_a$ represents a noise instrument and $v\in[0,1]$ is the visibility. We focus on three standard noise models:  dephasing noise,  white noise and worst-case noise. Respectively, these  correspond  to choosing $\eta_a^\text{noise}=\frac{\ketbra{00}+\ketbra{11}}{4}$, $\eta_a^\text{noise}=\frac{\mathds{1}}{8}$ and optimising over $\eta_a^\text{noise}$. The former corresponds to outputting a random $a$ and then emitting the state $\ket{a}$. The second also outputs a random $a$ but  emits the maximally mixed state. In contrast, since the final noise form makes the simulation  as powerful as possible, we refer to it as worst-case noise.

For all three noise-types, we have determined analytically the critical visibility for PI-simulation (see Appendix \ref{app:weakz}). The optimality of these results is verified by evaluating the SDP. The three critical visibilities are illustrated as  	solid lines in Fig~\ref{weakZ}. We see that sizeable amounts of noise are tolerated before simulation is possible, even in the worst-case setting. This attests to robust non-projective advantages. For dephasing and worst-case noise, the most noise-tolerant instrument corresponds to $\gamma=\frac{1}{\sqrt{2}}$, while for white noise it is somewhat smaller. In particular, the shape of the curves highlights that instruments that are close to sharp measurements ($\gamma=1$) are harder to simulate than instruments that are close to non-interacting measurements ($\gamma=0$).

Complementary to the above, it is interesting to consider instruments associated with extremal POVMs. These POVMs require more than two outcomes since otherwise they are always projective simulable. A prominent example is the qubit symmetric informationally complete POVM \cite{Renes2004}. It reads $\{\frac{1}{2}\ketbra{\varphi_a}\}_{a=1}^4$, corresponding to Bloch vectors  $[(1,1,1),(1,-1,-1),(-1,1,-1),(-1,-1,1)]/\sqrt{3}$. When this POVM is mixed with the random-output POVM, $\{\frac{\mathds{1}}{4}\}_a$, the critical visibility for projective simulation is  $v=\sqrt{2/3}\approx 0.817$ \cite{Oszmaniec2017}. We now consider the associated instrument, defined by Kraus operators $K_a=\frac{1}{\sqrt{2}}\varphi_a$ and mix it with a noise instrument, as in \eqref{noise}. The noise can be represented in more than one way; for instance both dephasing- and white-noise instruments  realise the POVM $\{\frac{\mathds{1}}{4}\}_a$. However, the critical visibility for instrument simulation is not the same. For the former, we find again $v=\sqrt{2/3}$ while for the latter we find $v\approx 0.773$. This showcases the relevance of considering the full measurement process also for extremal POVMs.

\subsection{Information-disturbance relations}
It is well-known that extracting information from a state also induces a disturbance in it \cite{Fuchs1996, Buscemi2006}. However, one may expect that non-projective instruments induce smaller disturbances than PIs. By examining an appropriate information-disturbance trade-off, we show how to check whether an instrument, $\boldsymbol{\mathcal{I}}$, defies PI-simulation.

The experimenter prepares a set of states $\{\psi_x\}$, passes them through $\boldsymbol{\mathcal{I}}$, and measures the quantum output with some POVMs $\{N_{b|y}\}_b$ where $y$ indexes the choice of measurement. The statistics of this experiment becomes $p(a,b|x,y)=\tr(\mathcal{I}_a(\psi_x)N_{b|y})$. To detect  non-projective behaviour, consider a linear witness
\begin{equation}\label{witness}
	W\equiv \sum_{a,b,x,y} c_{abxy} p(a,b|x,y)\leq \beta,
\end{equation}
where $c_{abxy}$ are real coefficients and $\beta$ is a bound satisfied by all PIs. We can  efficiently compute these bounds using our SDP relaxations for $\mathcal{P}$. This admits also further simplification: the linearity of $W$ lets us separately consider the witness for each rank-vector and then select the best result (see Appendix~\ref{AppHemisphere}).  If the experimenter observes a violation of the inequality \eqref{witness}, then no PI-simulation is possible.

\begin{figure}
	\centering
	\includegraphics[width=0.9\columnwidth]{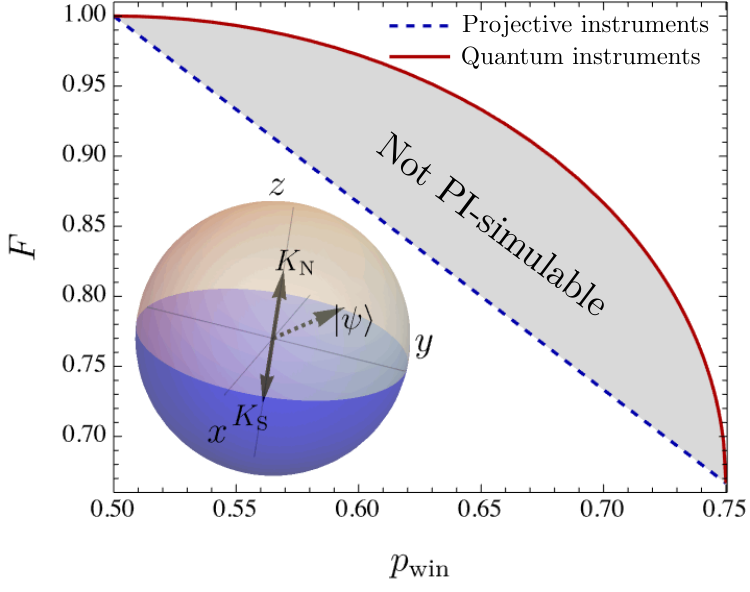}
	\caption{\textbf{Information-disturbance trade-off.} Fidelity between pre- and post-measurement state versus  success probability of hemisphere-discrimination. Inset: illustration of the hemisphere discrimination problem.}
	\label{hemis_dis}
\end{figure}

Let us consider a conceptually motivated example of such a situation, which pertains to all pure qubit states in quantum theory. For clarity, we explicity separate the information extraction part  from the state-disturbance part. Let $\psi$ be an arbitrary pure qubit state. Firstly, we want to determine if it belongs to the northern (N) or southern (S) Bloch-hemisphere. The success probability becomes $p_\text{win}=\frac{1}{2}\int_Nd\psi\tr(\mathcal{I}_\text{N}(\psi))+\frac{1}{2}\int_Sd\psi\tr(\mathcal{I}_\text{S}(\psi))$, where the outcomes are labelled $a\in\{\text{N,S}\}$. Secondly, we want to extract this information while keeping the state as little disturbed as possible. We quantify this through the fidelity of the post-measurement state, which corresponds to measuring $\{N_{b|\psi}\}_b=\{\ketbra{\psi},\openone-\ketbra{\psi}\}$. The average fidelity is $F=\int d\psi \bracket{\psi}{\sum_a \mathcal{I}_a(\psi)}{\psi}$. The optimal trade-off between $F$ and $p_\text{win}$ depends on whether we use generic quantum instruments (Q) or PIs On the one hand, we can select our instrument as the unsharp Pauli-$Z$ observable previously discussed. On the other hand, we can select it as any PI. As shown in Appendix~\ref{AppHemisphere}, the information-disturbance trade-offs become  
\begin{equation}\label{9a}
	 F_{\text{PI}}=\frac{5-4p_\text{win}}{3}, \quad F_Q=\frac{2+\sqrt{16p_\text{win}(1-p_\text{win})-3}}{3},
\end{equation}
where $p_\text{win}\in[1/2,3/4]$ ranges from its trivial value (random guess) to the maximal value possible. We see that non-projective instruments can better preserve the fidelity than can the PIs;  the latters are restricted to a linear trade-off (see Fig~\ref{hemis_dis}). Hence, a better-than-linear trade-off implies that the instrument is genuinely non-projective.

\subsection{Sequential Bell nonlocality} 

There has recently been much interest in sequential violations of Bell inequalities \cite{Cai2025}. In this scenario, Alice and Bob share a two-qubit entangled state $\Psi$ and aim to violate the CHSH inequality, which reads $\mathcal{S}_\text{AB}\equiv \expect{A_0B_0+A_0B_1+A_1B_0-A_1B_1}_\Psi\leq 2$ where $A_x$ and $B_y$ are Alice's and Bob's observables. The average state after Bob's measurement is $\Psi_\text{post}=\frac{1}{2}\sum_{b,y}(\openone\otimes K_{b|y})\Psi(\openone\otimes K_{b|y}^\dagger)$, where $K_{b|y}$ are the Kraus operators of Bob's  $y$'th instrument. Bob's share of $\Psi_\text{post}$ is relayed to Charlie who measures it in another CHSH test with Alice; $\mathcal{S}_\text{AC}\equiv \expect{A_0C_0+A_0C_1+A_1C_0-A_1C_1}_{\Psi_\text{post}}$, where $C_z$ is Charlie's observable. The goal is to achieve a double violation, namely $\mathcal{S}_\text{AB}>2$ and $\mathcal{S}_\text{AC}>2$. Thus, Bob's instrument must measure  strongly enough to violate the inequality, but weakly enough to make a violation possible also for Charlie. Non-projective instruments are natural for this task \cite{Silva2015} but it  was recently found that PIs suffice; they achieve   $\mathcal{S}_\text{AB}=\mathcal{S}_\text{AC}\approx 2.046$ \cite{Steffinlongo2022}.

\begin{figure}
	\centering
	\includegraphics[width=0.9\columnwidth]{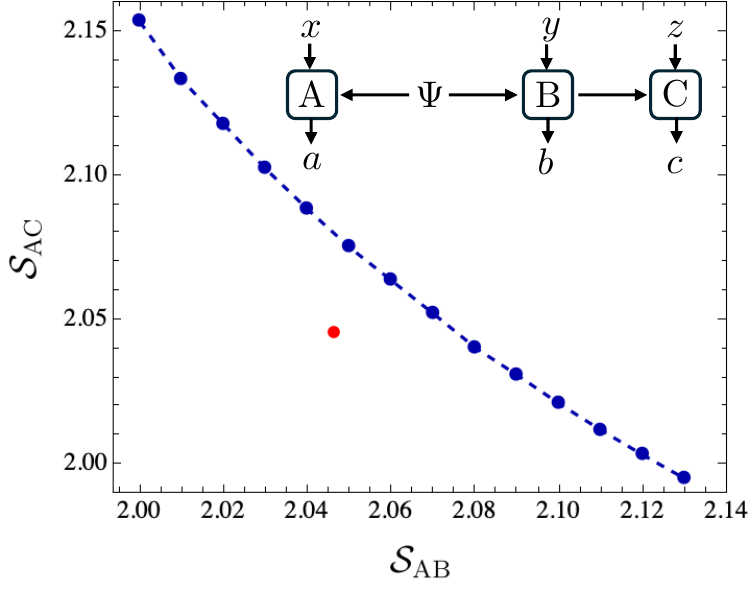}
	\caption{\textbf{Sequential nonlocality.} Numerical trade-off between CHSH parameters $\mathcal{S}_\text{AB}$ and $\mathcal{S}_\text{AC}$ for PIs. Red dot is reported in \cite{Steffinlongo2022}. Inset: an illustration of the sequential CHSH scenario.}
	\label{doublechsh}
\end{figure}

We use our SDP characterisation of qubit PIs to show that larger double-violations are possible than previously known, albeit still weaker than what is possible with non-projective operations. For this, we have employed an alternating convex search routine \cite{Tavakoli2024} in which we iteratively optimise over (i) the qubit assemblage prepared by Alice for Bob, (ii) the PI instruments of Bob, and (iii) the qubit measurements of Charlie. See Appendix \ref{app:SeqBell} for details and \cite{OurCode} for our implementation. This procedure returns an explicit quantum model for the double violations. The results are illustrated in Fig.~\ref{doublechsh} together with the previously known double violation (red dot) from Ref.~\cite{Steffinlongo2022}. We find significantly enhanced violations, and a trade-off between the CHSH parameters that is convex. Notably, we find that the best two-qubit entangled state is not maximally entangled. This is an important reason for improving on Ref.~\cite{Steffinlongo2022}, which restricts its analysis to maximally entangled states. This restriction may appear reasonable since this state indeed is optimal when considering generic quantum instruments \cite{Bowles2020}, but as we have shown it is unjustified for PIs.

\section{High dimensionality}
We now turn to instruments that operate on system of higher-than-qubit dimension and ask how the advantages of non-projectivity depend on the dimension $d$. We focus on L\"uders instruments for an unsharp measurement of the $d$-dimensional computational basis. Its Kraus operators are
\begin{equation}\label{inst}
	K_a=\sqrt{\frac{1+\gamma}{2}}\ketbra{a}+\sqrt{\frac{1-\gamma}{2(d-1)}}\left(\openone-\ketbra{a}\right),
\end{equation}
where $\gamma\in[0,1]$ is the sharpness parameter. The coefficient in front of  $\openone-\ketbra{a}$ is fixed by normalisation.  For $d=2$, these are the instruments considered in section~\ref{secX}.

We begin with the case of qutrits ($d=3$). As in section~\ref{secX}, we consider the mixture of the  instrument with noise of the dephasing, white and worst-case type respectively; see Eq.~\eqref{noise}. Although we do not have a complete characterisation of $\mathcal{P}$ for qutrits, we find that our  SDP relaxation of this set is sufficient to reveal significant noise-advantages over the qubit case. This is shown by the dashed lines in Fig~\ref{weakZ}. These enhanced results motivate us to ask whether the advantage of non-projective instruments over PIs is scalable for large $d$. To answer this, we consider the instrument \eqref{inst} for arbitrary $d$ and study its mixture with noise. We focus on two types of noise, namely dephasing and worst-case noise. Considering both turns out to be important because they yield sharply contrasting results. Let us start with the former.

Let $\eta_a=(K_a\otimes\openone)[\phi^+](K_a^\dagger\otimes \openone)$ be the Choi representation of the instrument \eqref{inst}, let $\eta_a^\text{noise}=\frac{1}{d^2}\sum_{i=0}^{d-1}\ketbra{ii}$ be the completely dephased instrument and consider their mixture, as in Eq.~\eqref{noise}. In Appendix~\ref{AppHighDim} we prove the following upper bound on the critical dephasing-visibility required for a PI-simulation; 
\begin{align}\label{vcrit}
v_\text{deph}(d) \leq   \frac{2(d-1)}{d\left(1+\gamma+\sqrt{(1-\gamma^2)(d-1)}\right)-2}. 
\end{align}
This bound is tight when $\gamma\geq 1-\frac{2}{d}$ and appears to be near-optimal otherwise. Our proof is based on  using  the strong duality theorem of SDP. By considering the dual form of our SDP relaxation of $\mathcal{P}$, we can find upper bounds on $v_\text{deph}(d)$ by constructing feasible points. We have provided explicit such constructions.  Importantly, the critical visibility \eqref{vcrit} scales as $O(\frac{1}{\sqrt{d}})$ for any $\gamma\neq 1$. This  tends to zero for large $d$ and shows that  the non-projective feature of the instruments becomes unboundedly noise-robust for high dimensions.

However, such a favourable scaling is not to be expected from arbitrary types of noise. To show this, consider noise of the worst-case type, i.e.~$\{\eta_a^\text{noise}\}$ can be selected arbitrarily. We consider the following specific choice,
\begin{align}\nonumber
		\eta_a^{\text{worst}} = &x_1\ketbra{aa}{aa}+x_2\sum_{k\neq a} \sum_{l\neq a}\ketbra{kk}{ll}\\& -\sqrt{x_1 x_2} \sum_{j\neq a} (\ketbra{aa}{jj}+h.c),
\end{align}
for some suitable coefficient $x_1$. By normalisation, we have $x_2=(1-dx_1)/(d(d-1))$. In Appendix~\ref{AppWorst} we construct explicit PIs that simulate the corresponding noisy L\"uders instrument \eqref{noise} for any $d$ and $\gamma$. This analytical model precisely matches the curves for $d=2,3$ illustrated in Figure~\ref{weakZ}, suggesting that it may be optimal in general. The weakest simulation is found at $\gamma=1/\sqrt{d}$, at which the model returns
\begin{equation}
v_\text{worst}(d)\geq \frac{1}{2}\left(1+\frac{1}{\sqrt{d}}\right).
\end{equation}
In the limit of large $d$, this converges to $\frac{1}{2}$. Thus, in any dimension and for any sharpness, this type of noise makes impossible a non-projective advantage that is more than $50\%$ noise-tolerance. This shows that particularly detrimental forms of noise can lead to powerful projective simulations even in high-dimensional systems.

\section{Projective simulation of POVMs}
By simply ignoring the quantum output, our approach to projectively simulating quantum instruments reduces to a projective simulation of POVMs. Since this problem has received prior interest in theory \cite{Oszmaniec2017, Guerini2017, Tavakoli2019, Tavakoli2020a} and experiment \cite{Gomez2016, Tavakoli2020, Smania2020, Martinez2023, Wang2023, Feng2023} we make explicit how our methods apply to it as a special case. Let $\{M_a\}$ be a POVM on $\mathcal{H}_A$. In a projective simulation, we aim to decompose it as $M_a=\sum_{\lambda} q_\lambda E_{a|\lambda}$ where $\{E_{a|\lambda}\}$ are projective measurements. Note that post-processing can be neglected without loss of generality \cite{Oszmaniec2017}. As before, we partition the space of projective measurements via their rank-vectors $\vec{r}$. This leads to the following characterisation
\begin{align}\nonumber
&	M_a=\sum_{\vec{r}} F_{a|\vec{r}}, \qquad \text{where} \quad F_{a|\vec{r}}\succeq 0,\\
&	\tr(F_{a|\vec{r}})=q_{\vec{r}}\ r_a, \qquad\qquad \sum_a F_{a|\vec{r}}=q_{\vec{r}}\openone.
\end{align}
where $q_{\vec{r}}=\frac{1}{d}\sum_a \tr(F_{a|\vec{r}})$ is the probability of selecting the rank-vector $\vec{r}$. This can be viewed as an SDP relaxation of the set of projective simulable measurements and it is obtained directly from Theorem~\ref{thm1} by reducing the Choi operators to the space $\mathcal{H}_A$. This gives a computationally simple way to prove the failure of projective measurement simulability.

\section{Conclusion}
We have introduced  projective quantum instruments. This class of instruments is motivated by the operational limitations of an experimenter who attemtps to perform a measurement process using only projective measurements on the incoming system and quantum processing of the post-projection state.  One can view this as a resource theory premiss, in which experiments do not have access to non-projective measurements when realising quantum instruments.

We have developed general and computationally efficient methods for characterising the quantum instruments that admit a projective realisation.  For qubits, we found a complete characterisation which allowed us solve several quantum information tasks. We also showed that  instrument acting on high-dimensional quantum systems can under relevant noise conditions have scalable advantages over projective instruments.  This indicates that dimensionality is a powerful resource for quantum instruments. Interestingly, the analogous has been shown not to hold for POVMs: Ref.~\cite{Kotowski2025} shows that there exists a non-zero  isotropic-noise visibility at which a projective measurement simulation is possible for any POVM in any dimension. This highlights that considering the full measurement process leads to sizably stronger advantages from non-projectivity than considering only the classical outcome of the same process.  A simple intuition for the enhanced scaling encountered for instruments stems from our reported  connection between their simulability and an entanglement classification problem. It is well-known that entanglement under several standard noise-conditions can uphold scalable robustness as the dimension increases. We have seen that this qualitative feature carries over to the advantages of non-projective instruments over projective instruments. While we also showed that non-projective advantages are not always scalable, this is to be expected due to the tailored selection of noise that is maximally detrimental for the considered  instrument. Such noise is typically not physically natural.

\begin{acknowledgments}
We thank Micha{\l} Oszmaniec for feedback. This work is supported by the Knut and Alice Wallenberg Foundation through the Wallenberg Center for Quantum Technology (WACQT) and the Swedish Research Council under Contract No.~2023-03498. S.K. acknowledges support from the Swiss National Science Foundation Grant No.
P500PT\textunderscore222265.
\end{acknowledgments}

\bibliography{bib_instruments}

\begin{thebibliography}{42}%
\makeatletter
\providecommand \@ifxundefined [1]{%
 \@ifx{#1\undefined}
}%
\providecommand \@ifnum [1]{%
 \ifnum #1\expandafter \@firstoftwo
 \else \expandafter \@secondoftwo
 \fi
}%
\providecommand \@ifx [1]{%
 \ifx #1\expandafter \@firstoftwo
 \else \expandafter \@secondoftwo
 \fi
}%
\providecommand \natexlab [1]{#1}%
\providecommand \enquote  [1]{``#1''}%
\providecommand \bibnamefont  [1]{#1}%
\providecommand \bibfnamefont [1]{#1}%
\providecommand \citenamefont [1]{#1}%
\providecommand \href@noop [0]{\@secondoftwo}%
\providecommand \href [0]{\begingroup \@sanitize@url \@href}%
\providecommand \@href[1]{\@@startlink{#1}\@@href}%
\providecommand \@@href[1]{\endgroup#1\@@endlink}%
\providecommand \@sanitize@url [0]{\catcode `\\12\catcode `\$12\catcode
  `\&12\catcode `\#12\catcode `\^12\catcode `\_12\catcode `\%12\relax}%
\providecommand \@@startlink[1]{}%
\providecommand \@@endlink[0]{}%
\providecommand \url  [0]{\begingroup\@sanitize@url \@url }%
\providecommand \@url [1]{\endgroup\@href {#1}{\urlprefix }}%
\providecommand \urlprefix  [0]{URL }%
\providecommand \Eprint [0]{\href }%
\providecommand \doibase [0]{https://doi.org/}%
\providecommand \selectlanguage [0]{\@gobble}%
\providecommand \bibinfo  [0]{\@secondoftwo}%
\providecommand \bibfield  [0]{\@secondoftwo}%
\providecommand \translation [1]{[#1]}%
\providecommand \BibitemOpen [0]{}%
\providecommand \bibitemStop [0]{}%
\providecommand \bibitemNoStop [0]{.\EOS\space}%
\providecommand \EOS [0]{\spacefactor3000\relax}%
\providecommand \BibitemShut  [1]{\csname bibitem#1\endcsname}%
\let\auto@bib@innerbib\@empty
%</preamble>
\bibitem [{\citenamefont {Chitambar}\ and\ \citenamefont
  {Gour}(2019)}]{Chitambar2019}%
  \BibitemOpen
  \bibfield  {author} {\bibinfo {author} {\bibfnamefont {E.}~\bibnamefont
  {Chitambar}}\ and\ \bibinfo {author} {\bibfnamefont {G.}~\bibnamefont
  {Gour}},\ }\bibfield  {title} {\bibinfo {title} {Quantum resource theories},\
  }\href {https://doi.org/10.1103/RevModPhys.91.025001} {\bibfield  {journal}
  {\bibinfo  {journal} {Rev. Mod. Phys.}\ }\textbf {\bibinfo {volume} {91}},\
  \bibinfo {pages} {025001} (\bibinfo {year} {2019})}\BibitemShut {NoStop}%
\bibitem [{\citenamefont {Oszmaniec}\ \emph {et~al.}(2017)\citenamefont
  {Oszmaniec}, \citenamefont {Guerini}, \citenamefont {Wittek},\ and\
  \citenamefont {Ac\'{\i}n}}]{Oszmaniec2017}%
  \BibitemOpen
  \bibfield  {author} {\bibinfo {author} {\bibfnamefont {M.}~\bibnamefont
  {Oszmaniec}}, \bibinfo {author} {\bibfnamefont {L.}~\bibnamefont {Guerini}},
  \bibinfo {author} {\bibfnamefont {P.}~\bibnamefont {Wittek}},\ and\ \bibinfo
  {author} {\bibfnamefont {A.}~\bibnamefont {Ac\'{\i}n}},\ }\bibfield  {title}
  {\bibinfo {title} {Simulating positive-operator-valued measures with
  projective measurements},\ }\href
  {https://doi.org/10.1103/PhysRevLett.119.190501} {\bibfield  {journal}
  {\bibinfo  {journal} {Phys. Rev. Lett.}\ }\textbf {\bibinfo {volume} {119}},\
  \bibinfo {pages} {190501} (\bibinfo {year} {2017})}\BibitemShut {NoStop}%
\bibitem [{\citenamefont {G\'omez}\ \emph {et~al.}(2016)\citenamefont
  {G\'omez}, \citenamefont {G\'omez}, \citenamefont {Gonz\'alez}, \citenamefont
  {Ca\~nas}, \citenamefont {Barra}, \citenamefont {Delgado}, \citenamefont
  {Xavier}, \citenamefont {Cabello}, \citenamefont {Kleinmann}, \citenamefont
  {V\'ertesi},\ and\ \citenamefont {Lima}}]{Gomez2016}%
  \BibitemOpen
  \bibfield  {author} {\bibinfo {author} {\bibfnamefont {E.~S.}\ \bibnamefont
  {G\'omez}}, \bibinfo {author} {\bibfnamefont {S.}~\bibnamefont {G\'omez}},
  \bibinfo {author} {\bibfnamefont {P.}~\bibnamefont {Gonz\'alez}}, \bibinfo
  {author} {\bibfnamefont {G.}~\bibnamefont {Ca\~nas}}, \bibinfo {author}
  {\bibfnamefont {J.~F.}\ \bibnamefont {Barra}}, \bibinfo {author}
  {\bibfnamefont {A.}~\bibnamefont {Delgado}}, \bibinfo {author} {\bibfnamefont
  {G.~B.}\ \bibnamefont {Xavier}}, \bibinfo {author} {\bibfnamefont
  {A.}~\bibnamefont {Cabello}}, \bibinfo {author} {\bibfnamefont
  {M.}~\bibnamefont {Kleinmann}}, \bibinfo {author} {\bibfnamefont
  {T.}~\bibnamefont {V\'ertesi}},\ and\ \bibinfo {author} {\bibfnamefont
  {G.}~\bibnamefont {Lima}},\ }\bibfield  {title} {\bibinfo {title}
  {Device-independent certification of a nonprojective qubit measurement},\
  }\href {https://doi.org/10.1103/PhysRevLett.117.260401} {\bibfield  {journal}
  {\bibinfo  {journal} {Phys. Rev. Lett.}\ }\textbf {\bibinfo {volume} {117}},\
  \bibinfo {pages} {260401} (\bibinfo {year} {2016})}\BibitemShut {NoStop}%
\bibitem [{\citenamefont {Tavakoli}\ \emph {et~al.}(2020)\citenamefont
  {Tavakoli}, \citenamefont {Smania}, \citenamefont {Vértesi}, \citenamefont
  {Brunner},\ and\ \citenamefont {Bourennane}}]{Tavakoli2020}%
  \BibitemOpen
  \bibfield  {author} {\bibinfo {author} {\bibfnamefont {A.}~\bibnamefont
  {Tavakoli}}, \bibinfo {author} {\bibfnamefont {M.}~\bibnamefont {Smania}},
  \bibinfo {author} {\bibfnamefont {T.}~\bibnamefont {Vértesi}}, \bibinfo
  {author} {\bibfnamefont {N.}~\bibnamefont {Brunner}},\ and\ \bibinfo {author}
  {\bibfnamefont {M.}~\bibnamefont {Bourennane}},\ }\bibfield  {title}
  {\bibinfo {title} {Self-testing nonprojective quantum measurements in
  prepare-and-measure experiments},\ }\href
  {https://doi.org/10.1126/sciadv.aaw6664} {\bibfield  {journal} {\bibinfo
  {journal} {Science Advances}\ }\textbf {\bibinfo {volume} {6}},\ \bibinfo
  {pages} {eaaw6664} (\bibinfo {year} {2020})}\BibitemShut {NoStop}%
\bibitem [{\citenamefont {Smania}\ \emph {et~al.}(2020)\citenamefont {Smania},
  \citenamefont {Mironowicz}, \citenamefont {Nawareg}, \citenamefont
  {Paw{\l}owski}, \citenamefont {Cabello},\ and\ \citenamefont
  {Bourennane}}]{Smania2020}%
  \BibitemOpen
  \bibfield  {author} {\bibinfo {author} {\bibfnamefont {M.}~\bibnamefont
  {Smania}}, \bibinfo {author} {\bibfnamefont {P.}~\bibnamefont {Mironowicz}},
  \bibinfo {author} {\bibfnamefont {M.}~\bibnamefont {Nawareg}}, \bibinfo
  {author} {\bibfnamefont {M.}~\bibnamefont {Paw{\l}owski}}, \bibinfo {author}
  {\bibfnamefont {A.}~\bibnamefont {Cabello}},\ and\ \bibinfo {author}
  {\bibfnamefont {M.}~\bibnamefont {Bourennane}},\ }\bibfield  {title}
  {\bibinfo {title} {Experimental certification of an informationally complete
  quantum measurement in a device-independent protocol},\ }\href
  {https://doi.org/10.1364/OPTICA.377959} {\bibfield  {journal} {\bibinfo
  {journal} {Optica}\ }\textbf {\bibinfo {volume} {7}},\ \bibinfo {pages} {123}
  (\bibinfo {year} {2020})}\BibitemShut {NoStop}%
\bibitem [{\citenamefont {Mart{\'i}nez}\ \emph {et~al.}(2023)\citenamefont
  {Mart{\'i}nez}, \citenamefont {G{\'o}mez}, \citenamefont {Cari{\~{n}}e},
  \citenamefont {Pereira}, \citenamefont {Delgado}, \citenamefont {Walborn},
  \citenamefont {Tavakoli},\ and\ \citenamefont {Lima}}]{Martinez2023}%
  \BibitemOpen
  \bibfield  {author} {\bibinfo {author} {\bibfnamefont {D.}~\bibnamefont
  {Mart{\'i}nez}}, \bibinfo {author} {\bibfnamefont {E.~S.}\ \bibnamefont
  {G{\'o}mez}}, \bibinfo {author} {\bibfnamefont {J.}~\bibnamefont
  {Cari{\~{n}}e}}, \bibinfo {author} {\bibfnamefont {L.}~\bibnamefont
  {Pereira}}, \bibinfo {author} {\bibfnamefont {A.}~\bibnamefont {Delgado}},
  \bibinfo {author} {\bibfnamefont {S.~P.}\ \bibnamefont {Walborn}}, \bibinfo
  {author} {\bibfnamefont {A.}~\bibnamefont {Tavakoli}},\ and\ \bibinfo
  {author} {\bibfnamefont {G.}~\bibnamefont {Lima}},\ }\bibfield  {title}
  {\bibinfo {title} {Certification of a non-projective qudit measurement using
  multiport beamsplitters},\ }\href
  {https://doi.org/10.1038/s41567-022-01845-z} {\bibfield  {journal} {\bibinfo
  {journal} {Nature Physics}\ }\textbf {\bibinfo {volume} {19}},\ \bibinfo
  {pages} {190} (\bibinfo {year} {2023})}\BibitemShut {NoStop}%
\bibitem [{\citenamefont {Wang}\ \emph {et~al.}(2023)\citenamefont {Wang},
  \citenamefont {Zhan}, \citenamefont {Li}, \citenamefont {Xiao}, \citenamefont
  {Zhu}, \citenamefont {Qu}, \citenamefont {Lin}, \citenamefont {Yu},\ and\
  \citenamefont {Xue}}]{Wang2023}%
  \BibitemOpen
  \bibfield  {author} {\bibinfo {author} {\bibfnamefont {X.}~\bibnamefont
  {Wang}}, \bibinfo {author} {\bibfnamefont {X.}~\bibnamefont {Zhan}}, \bibinfo
  {author} {\bibfnamefont {Y.}~\bibnamefont {Li}}, \bibinfo {author}
  {\bibfnamefont {L.}~\bibnamefont {Xiao}}, \bibinfo {author} {\bibfnamefont
  {G.}~\bibnamefont {Zhu}}, \bibinfo {author} {\bibfnamefont {D.}~\bibnamefont
  {Qu}}, \bibinfo {author} {\bibfnamefont {Q.}~\bibnamefont {Lin}}, \bibinfo
  {author} {\bibfnamefont {Y.}~\bibnamefont {Yu}},\ and\ \bibinfo {author}
  {\bibfnamefont {P.}~\bibnamefont {Xue}},\ }\bibfield  {title} {\bibinfo
  {title} {Generalized quantum measurements on a higher-dimensional system via
  quantum walks},\ }\href {https://doi.org/10.1103/PhysRevLett.131.150803}
  {\bibfield  {journal} {\bibinfo  {journal} {Phys. Rev. Lett.}\ }\textbf
  {\bibinfo {volume} {131}},\ \bibinfo {pages} {150803} (\bibinfo {year}
  {2023})}\BibitemShut {NoStop}%
\bibitem [{\citenamefont {Feng}\ \emph {et~al.}(2023)\citenamefont {Feng},
  \citenamefont {Hu}, \citenamefont {Zhang}, \citenamefont {Cheng},
  \citenamefont {Zhang}, \citenamefont {Guo}, \citenamefont {Ding},
  \citenamefont {Hou}, \citenamefont {Sun}, \citenamefont {Guo}, \citenamefont
  {Dai}, \citenamefont {Tavakoli}, \citenamefont {Ren},\ and\ \citenamefont
  {Liu}}]{Feng2023}%
  \BibitemOpen
  \bibfield  {author} {\bibinfo {author} {\bibfnamefont {L.-T.}\ \bibnamefont
  {Feng}}, \bibinfo {author} {\bibfnamefont {X.-M.}\ \bibnamefont {Hu}},
  \bibinfo {author} {\bibfnamefont {M.}~\bibnamefont {Zhang}}, \bibinfo
  {author} {\bibfnamefont {Y.-J.}\ \bibnamefont {Cheng}}, \bibinfo {author}
  {\bibfnamefont {C.}~\bibnamefont {Zhang}}, \bibinfo {author} {\bibfnamefont
  {Y.}~\bibnamefont {Guo}}, \bibinfo {author} {\bibfnamefont {Y.-Y.}\
  \bibnamefont {Ding}}, \bibinfo {author} {\bibfnamefont {Z.}~\bibnamefont
  {Hou}}, \bibinfo {author} {\bibfnamefont {F.-W.}\ \bibnamefont {Sun}},
  \bibinfo {author} {\bibfnamefont {G.-C.}\ \bibnamefont {Guo}}, \bibinfo
  {author} {\bibfnamefont {D.-X.}\ \bibnamefont {Dai}}, \bibinfo {author}
  {\bibfnamefont {A.}~\bibnamefont {Tavakoli}}, \bibinfo {author}
  {\bibfnamefont {X.-F.}\ \bibnamefont {Ren}},\ and\ \bibinfo {author}
  {\bibfnamefont {B.-H.}\ \bibnamefont {Liu}},\ }\href
  {https://arxiv.org/abs/2310.08838} {\bibinfo {title} {Higher-dimensional
  symmetric informationally complete measurement via programmable photonic
  integrated optics}} (\bibinfo {year} {2023}),\ \Eprint
  {https://arxiv.org/abs/2310.08838} {arXiv:2310.08838 [quant-ph]} \BibitemShut
  {NoStop}%
\bibitem [{\citenamefont {Busch}\ \emph {et~al.}(2016)\citenamefont {Busch},
  \citenamefont {Lahti}, \citenamefont {Pellonp\"a\"a},\ and\ \citenamefont
  {Ylinen}}]{QMeasurement}%
  \BibitemOpen
  \bibfield  {author} {\bibinfo {author} {\bibfnamefont {P.}~\bibnamefont
  {Busch}}, \bibinfo {author} {\bibfnamefont {P.}~\bibnamefont {Lahti}},
  \bibinfo {author} {\bibfnamefont {J.-P.}\ \bibnamefont {Pellonp\"a\"a}},\
  and\ \bibinfo {author} {\bibfnamefont {K.}~\bibnamefont {Ylinen}},\
  }\href@noop {} {\emph {\bibinfo {title} {Quantum Measurement}}}\ (\bibinfo
  {publisher} {Springer},\ \bibinfo {year} {2016})\BibitemShut {NoStop}%
\bibitem [{\citenamefont {Andersen}\ \emph {et~al.}(2006)\citenamefont
  {Andersen}, \citenamefont {Sabuncu}, \citenamefont {Filip},\ and\
  \citenamefont {Leuchs}}]{Andersen2006}%
  \BibitemOpen
  \bibfield  {author} {\bibinfo {author} {\bibfnamefont {U.~L.}\ \bibnamefont
  {Andersen}}, \bibinfo {author} {\bibfnamefont {M.}~\bibnamefont {Sabuncu}},
  \bibinfo {author} {\bibfnamefont {R.}~\bibnamefont {Filip}},\ and\ \bibinfo
  {author} {\bibfnamefont {G.}~\bibnamefont {Leuchs}},\ }\bibfield  {title}
  {\bibinfo {title} {Experimental demonstration of coherent state estimation
  with minimal disturbance},\ }\href
  {https://doi.org/10.1103/PhysRevLett.96.020409} {\bibfield  {journal}
  {\bibinfo  {journal} {Phys. Rev. Lett.}\ }\textbf {\bibinfo {volume} {96}},\
  \bibinfo {pages} {020409} (\bibinfo {year} {2006})}\BibitemShut {NoStop}%
\bibitem [{\citenamefont {Lim}\ \emph {et~al.}(2014)\citenamefont {Lim},
  \citenamefont {Ra}, \citenamefont {Hong}, \citenamefont {Lee},\ and\
  \citenamefont {Kim}}]{Lim2014}%
  \BibitemOpen
  \bibfield  {author} {\bibinfo {author} {\bibfnamefont {H.-T.}\ \bibnamefont
  {Lim}}, \bibinfo {author} {\bibfnamefont {Y.-S.}\ \bibnamefont {Ra}},
  \bibinfo {author} {\bibfnamefont {K.-H.}\ \bibnamefont {Hong}}, \bibinfo
  {author} {\bibfnamefont {S.-W.}\ \bibnamefont {Lee}},\ and\ \bibinfo {author}
  {\bibfnamefont {Y.-H.}\ \bibnamefont {Kim}},\ }\bibfield  {title} {\bibinfo
  {title} {Fundamental bounds in measurements for estimating quantum states},\
  }\href {https://doi.org/10.1103/PhysRevLett.113.020504} {\bibfield  {journal}
  {\bibinfo  {journal} {Phys. Rev. Lett.}\ }\textbf {\bibinfo {volume} {113}},\
  \bibinfo {pages} {020504} (\bibinfo {year} {2014})}\BibitemShut {NoStop}%
\bibitem [{\citenamefont {Schiavon}\ \emph {et~al.}(2017)\citenamefont
  {Schiavon}, \citenamefont {Calderaro}, \citenamefont {Pittaluga},
  \citenamefont {Vallone},\ and\ \citenamefont {Villoresi}}]{Schiavon2017}%
  \BibitemOpen
  \bibfield  {author} {\bibinfo {author} {\bibfnamefont {M.}~\bibnamefont
  {Schiavon}}, \bibinfo {author} {\bibfnamefont {L.}~\bibnamefont {Calderaro}},
  \bibinfo {author} {\bibfnamefont {M.}~\bibnamefont {Pittaluga}}, \bibinfo
  {author} {\bibfnamefont {G.}~\bibnamefont {Vallone}},\ and\ \bibinfo {author}
  {\bibfnamefont {P.}~\bibnamefont {Villoresi}},\ }\bibfield  {title} {\bibinfo
  {title} {Three-observer bell inequality violation on a two-qubit entangled
  state},\ }\href {https://doi.org/10.1088/2058-9565/aa62be} {\bibfield
  {journal} {\bibinfo  {journal} {Quantum Science and Technology}\ }\textbf
  {\bibinfo {volume} {2}},\ \bibinfo {pages} {015010} (\bibinfo {year}
  {2017})}\BibitemShut {NoStop}%
\bibitem [{\citenamefont {Anwer}\ \emph {et~al.}(2020)\citenamefont {Anwer},
  \citenamefont {Muhammad}, \citenamefont {Cherifi}, \citenamefont {Miklin},
  \citenamefont {Tavakoli},\ and\ \citenamefont {Bourennane}}]{Anwer2020}%
  \BibitemOpen
  \bibfield  {author} {\bibinfo {author} {\bibfnamefont {H.}~\bibnamefont
  {Anwer}}, \bibinfo {author} {\bibfnamefont {S.}~\bibnamefont {Muhammad}},
  \bibinfo {author} {\bibfnamefont {W.}~\bibnamefont {Cherifi}}, \bibinfo
  {author} {\bibfnamefont {N.}~\bibnamefont {Miklin}}, \bibinfo {author}
  {\bibfnamefont {A.}~\bibnamefont {Tavakoli}},\ and\ \bibinfo {author}
  {\bibfnamefont {M.}~\bibnamefont {Bourennane}},\ }\bibfield  {title}
  {\bibinfo {title} {Experimental characterization of unsharp qubit observables
  and sequential measurement incompatibility via quantum random access codes},\
  }\href {https://doi.org/10.1103/PhysRevLett.125.080403} {\bibfield  {journal}
  {\bibinfo  {journal} {Phys. Rev. Lett.}\ }\textbf {\bibinfo {volume} {125}},\
  \bibinfo {pages} {080403} (\bibinfo {year} {2020})}\BibitemShut {NoStop}%
\bibitem [{\citenamefont {Foletto}\ \emph {et~al.}(2020)\citenamefont
  {Foletto}, \citenamefont {Calderaro}, \citenamefont {Vallone},\ and\
  \citenamefont {Villoresi}}]{Foletto2020}%
  \BibitemOpen
  \bibfield  {author} {\bibinfo {author} {\bibfnamefont {G.}~\bibnamefont
  {Foletto}}, \bibinfo {author} {\bibfnamefont {L.}~\bibnamefont {Calderaro}},
  \bibinfo {author} {\bibfnamefont {G.}~\bibnamefont {Vallone}},\ and\ \bibinfo
  {author} {\bibfnamefont {P.}~\bibnamefont {Villoresi}},\ }\bibfield  {title}
  {\bibinfo {title} {Experimental demonstration of sequential quantum random
  access codes},\ }\href {https://doi.org/10.1103/PhysRevResearch.2.033205}
  {\bibfield  {journal} {\bibinfo  {journal} {Phys. Rev. Res.}\ }\textbf
  {\bibinfo {volume} {2}},\ \bibinfo {pages} {033205} (\bibinfo {year}
  {2020})}\BibitemShut {NoStop}%
\bibitem [{\citenamefont {Anwer}\ \emph {et~al.}(2021)\citenamefont {Anwer},
  \citenamefont {Wilson}, \citenamefont {Silva}, \citenamefont {Muhammad},
  \citenamefont {Tavakoli},\ and\ \citenamefont {Bourennane}}]{Anwer2021}%
  \BibitemOpen
  \bibfield  {author} {\bibinfo {author} {\bibfnamefont {H.}~\bibnamefont
  {Anwer}}, \bibinfo {author} {\bibfnamefont {N.}~\bibnamefont {Wilson}},
  \bibinfo {author} {\bibfnamefont {R.}~\bibnamefont {Silva}}, \bibinfo
  {author} {\bibfnamefont {S.}~\bibnamefont {Muhammad}}, \bibinfo {author}
  {\bibfnamefont {A.}~\bibnamefont {Tavakoli}},\ and\ \bibinfo {author}
  {\bibfnamefont {M.}~\bibnamefont {Bourennane}},\ }\bibfield  {title}
  {\bibinfo {title} {Noise-robust preparation contextuality shared between any
  number of observers via unsharp measurements},\ }\href
  {https://doi.org/10.22331/q-2021-09-28-551} {\bibfield  {journal} {\bibinfo
  {journal} {{Quantum}}\ }\textbf {\bibinfo {volume} {5}},\ \bibinfo {pages}
  {551} (\bibinfo {year} {2021})}\BibitemShut {NoStop}%
\bibitem [{\citenamefont {Tavakoli}\ \emph {et~al.}(2024)\citenamefont
  {Tavakoli}, \citenamefont {Pozas-Kerstjens}, \citenamefont {Brown},\ and\
  \citenamefont {Ara\'ujo}}]{Tavakoli2024}%
  \BibitemOpen
  \bibfield  {author} {\bibinfo {author} {\bibfnamefont {A.}~\bibnamefont
  {Tavakoli}}, \bibinfo {author} {\bibfnamefont {A.}~\bibnamefont
  {Pozas-Kerstjens}}, \bibinfo {author} {\bibfnamefont {P.}~\bibnamefont
  {Brown}},\ and\ \bibinfo {author} {\bibfnamefont {M.}~\bibnamefont
  {Ara\'ujo}},\ }\bibfield  {title} {\bibinfo {title} {Semidefinite programming
  relaxations for quantum correlations},\ }\href
  {https://doi.org/10.1103/RevModPhys.96.045006} {\bibfield  {journal}
  {\bibinfo  {journal} {Rev. Mod. Phys.}\ }\textbf {\bibinfo {volume} {96}},\
  \bibinfo {pages} {045006} (\bibinfo {year} {2024})}\BibitemShut {NoStop}%
\bibitem [{\citenamefont {Steffinlongo}\ and\ \citenamefont
  {Tavakoli}(2022)}]{Steffinlongo2022}%
  \BibitemOpen
  \bibfield  {author} {\bibinfo {author} {\bibfnamefont {A.}~\bibnamefont
  {Steffinlongo}}\ and\ \bibinfo {author} {\bibfnamefont {A.}~\bibnamefont
  {Tavakoli}},\ }\bibfield  {title} {\bibinfo {title} {Projective measurements
  are sufficient for recycling nonlocality},\ }\href
  {https://doi.org/10.1103/PhysRevLett.129.230402} {\bibfield  {journal}
  {\bibinfo  {journal} {Phys. Rev. Lett.}\ }\textbf {\bibinfo {volume} {129}},\
  \bibinfo {pages} {230402} (\bibinfo {year} {2022})}\BibitemShut {NoStop}%
\bibitem [{\citenamefont {Terhal}\ and\ \citenamefont
  {Horodecki}(2000)}]{Terhal2000}%
  \BibitemOpen
  \bibfield  {author} {\bibinfo {author} {\bibfnamefont {B.~M.}\ \bibnamefont
  {Terhal}}\ and\ \bibinfo {author} {\bibfnamefont {P.}~\bibnamefont
  {Horodecki}},\ }\bibfield  {title} {\bibinfo {title} {Schmidt number for
  density matrices},\ }\href
  {https://link.aps.org/doi/10.1103/PhysRevA.61.040301} {\bibfield  {journal}
  {\bibinfo  {journal} {Phys. Rev. A}\ }\textbf {\bibinfo {volume} {61}},\
  \bibinfo {pages} {040301} (\bibinfo {year} {2000})}\BibitemShut {NoStop}%
\bibitem [{\citenamefont {Horodecki}\ \emph {et~al.}(1996)\citenamefont
  {Horodecki}, \citenamefont {Horodecki},\ and\ \citenamefont
  {Horodecki}}]{Horodecki1996}%
  \BibitemOpen
  \bibfield  {author} {\bibinfo {author} {\bibfnamefont {M.}~\bibnamefont
  {Horodecki}}, \bibinfo {author} {\bibfnamefont {P.}~\bibnamefont
  {Horodecki}},\ and\ \bibinfo {author} {\bibfnamefont {R.}~\bibnamefont
  {Horodecki}},\ }\bibfield  {title} {\bibinfo {title} {Separability of mixed
  states: necessary and sufficient conditions},\ }\href
  {https://doi.org/https://doi.org/10.1016/S0375-9601(96)00706-2} {\bibfield
  {journal} {\bibinfo  {journal} {Physics Letters A}\ }\textbf {\bibinfo
  {volume} {223}},\ \bibinfo {pages} {1} (\bibinfo {year} {1996})}\BibitemShut
  {NoStop}%
\bibitem [{\citenamefont {Gharibian}(2010)}]{Gharibian2010}%
  \BibitemOpen
  \bibfield  {author} {\bibinfo {author} {\bibfnamefont {S.}~\bibnamefont
  {Gharibian}},\ }\bibfield  {title} {\bibinfo {title} {Strong {NP}-hardness of
  the quantum separability problem},\ }\href
  {https://doi.org/10.26421/QIC10.3-4-11} {\bibfield  {journal} {\bibinfo
  {journal} {Quantum Information and Computation}\ }\textbf {\bibinfo {volume}
  {10}},\ \bibinfo {pages} {343} (\bibinfo {year} {2010})}\BibitemShut
  {NoStop}%
\bibitem [{\citenamefont {Sperling}\ and\ \citenamefont
  {Vogel}(2011)}]{Sperling2011}%
  \BibitemOpen
  \bibfield  {author} {\bibinfo {author} {\bibfnamefont {J.}~\bibnamefont
  {Sperling}}\ and\ \bibinfo {author} {\bibfnamefont {W.}~\bibnamefont
  {Vogel}},\ }\bibfield  {title} {\bibinfo {title} {Determination of the
  schmidt number},\ }\href {https://doi.org/10.1103/PhysRevA.83.042315}
  {\bibfield  {journal} {\bibinfo  {journal} {Phys. Rev. A}\ }\textbf {\bibinfo
  {volume} {83}},\ \bibinfo {pages} {042315} (\bibinfo {year}
  {2011})}\BibitemShut {NoStop}%
\bibitem [{\citenamefont {Shahandeh}\ \emph {et~al.}(2014)\citenamefont
  {Shahandeh}, \citenamefont {Sperling},\ and\ \citenamefont
  {Vogel}}]{Shahandeh2014}%
  \BibitemOpen
  \bibfield  {author} {\bibinfo {author} {\bibfnamefont {F.}~\bibnamefont
  {Shahandeh}}, \bibinfo {author} {\bibfnamefont {J.}~\bibnamefont
  {Sperling}},\ and\ \bibinfo {author} {\bibfnamefont {W.}~\bibnamefont
  {Vogel}},\ }\bibfield  {title} {\bibinfo {title} {Structural quantification
  of entanglement},\ }\href {https://doi.org/10.1103/PhysRevLett.113.260502}
  {\bibfield  {journal} {\bibinfo  {journal} {Phys. Rev. Lett.}\ }\textbf
  {\bibinfo {volume} {113}},\ \bibinfo {pages} {260502} (\bibinfo {year}
  {2014})}\BibitemShut {NoStop}%
\bibitem [{\citenamefont {Weilenmann}\ \emph {et~al.}(2020)\citenamefont
  {Weilenmann}, \citenamefont {Dive}, \citenamefont {Trillo}, \citenamefont
  {Aguilar},\ and\ \citenamefont {Navascu\'es}}]{Weilenmann2020}%
  \BibitemOpen
  \bibfield  {author} {\bibinfo {author} {\bibfnamefont {M.}~\bibnamefont
  {Weilenmann}}, \bibinfo {author} {\bibfnamefont {B.}~\bibnamefont {Dive}},
  \bibinfo {author} {\bibfnamefont {D.}~\bibnamefont {Trillo}}, \bibinfo
  {author} {\bibfnamefont {E.~A.}\ \bibnamefont {Aguilar}},\ and\ \bibinfo
  {author} {\bibfnamefont {M.}~\bibnamefont {Navascu\'es}},\ }\bibfield
  {title} {\bibinfo {title} {Entanglement detection beyond measuring
  fidelities},\ }\href {https://doi.org/10.1103/PhysRevLett.124.200502}
  {\bibfield  {journal} {\bibinfo  {journal} {Phys. Rev. Lett.}\ }\textbf
  {\bibinfo {volume} {124}},\ \bibinfo {pages} {200502} (\bibinfo {year}
  {2020})}\BibitemShut {NoStop}%
\bibitem [{\citenamefont {Morelli}\ \emph {et~al.}(2023)\citenamefont
  {Morelli}, \citenamefont {Huber},\ and\ \citenamefont
  {Tavakoli}}]{Morelli2023}%
  \BibitemOpen
  \bibfield  {author} {\bibinfo {author} {\bibfnamefont {S.}~\bibnamefont
  {Morelli}}, \bibinfo {author} {\bibfnamefont {M.}~\bibnamefont {Huber}},\
  and\ \bibinfo {author} {\bibfnamefont {A.}~\bibnamefont {Tavakoli}},\
  }\bibfield  {title} {\bibinfo {title} {Resource-efficient high-dimensional
  entanglement detection via symmetric projections},\ }\href
  {https://doi.org/10.1103/PhysRevLett.131.170201} {\bibfield  {journal}
  {\bibinfo  {journal} {Phys. Rev. Lett.}\ }\textbf {\bibinfo {volume} {131}},\
  \bibinfo {pages} {170201} (\bibinfo {year} {2023})}\BibitemShut {NoStop}%
\bibitem [{\citenamefont {Tavakoli}\ and\ \citenamefont
  {Morelli}(2024)}]{Tavakoli2024b}%
  \BibitemOpen
  \bibfield  {author} {\bibinfo {author} {\bibfnamefont {A.}~\bibnamefont
  {Tavakoli}}\ and\ \bibinfo {author} {\bibfnamefont {S.}~\bibnamefont
  {Morelli}},\ }\bibfield  {title} {\bibinfo {title} {Enhanced schmidt-number
  criteria based on correlation trace norms},\ }\href
  {https://doi.org/10.1103/PhysRevA.110.062417} {\bibfield  {journal} {\bibinfo
   {journal} {Phys. Rev. A}\ }\textbf {\bibinfo {volume} {110}},\ \bibinfo
  {pages} {062417} (\bibinfo {year} {2024})}\BibitemShut {NoStop}%
\bibitem [{\citenamefont {Cobucci}\ and\ \citenamefont
  {Tavakoli}(2024)}]{Cobucci2024}%
  \BibitemOpen
  \bibfield  {author} {\bibinfo {author} {\bibfnamefont {G.}~\bibnamefont
  {Cobucci}}\ and\ \bibinfo {author} {\bibfnamefont {A.}~\bibnamefont
  {Tavakoli}},\ }\bibfield  {title} {\bibinfo {title} {Detecting the
  dimensionality of genuine multiparticle entanglement},\ }\href
  {https://doi.org/10.1126/sciadv.adq4467} {\bibfield  {journal} {\bibinfo
  {journal} {Science Advances}\ }\textbf {\bibinfo {volume} {10}},\ \bibinfo
  {pages} {eadq4467} (\bibinfo {year} {2024})}\BibitemShut {NoStop}%
\bibitem [{\citenamefont {Tomiyama}(1985)}]{Tomiyama1985}%
  \BibitemOpen
  \bibfield  {author} {\bibinfo {author} {\bibfnamefont {J.}~\bibnamefont
  {Tomiyama}},\ }\bibfield  {title} {\bibinfo {title} {On the geometry of
  positive maps in matrix algebras. {II}},\ }\href
  {https://doi.org/https://doi.org/10.1016/0024-3795(85)90074-6} {\bibfield
  {journal} {\bibinfo  {journal} {Linear Algebra and its Applications}\
  }\textbf {\bibinfo {volume} {69}},\ \bibinfo {pages} {169} (\bibinfo {year}
  {1985})}\BibitemShut {NoStop}%
\bibitem [{Our()}]{OurCode}%
  \BibitemOpen
  \href@noop {} {}\bibinfo {note}
  {\url{https://github.com/ShishirKhandelwal94/Simulating-quantum-instruments}}\BibitemShut
  {NoStop}%
\bibitem [{\citenamefont {Pellonp\"{a}\"{a}}(2012)}]{Pellonpaa2013}%
  \BibitemOpen
  \bibfield  {author} {\bibinfo {author} {\bibfnamefont {J.-P.}\ \bibnamefont
  {Pellonp\"{a}\"{a}}},\ }\bibfield  {title} {\bibinfo {title} {Quantum
  instruments: I. extreme instruments},\ }\href
  {https://doi.org/10.1088/1751-8113/46/2/025302} {\bibfield  {journal}
  {\bibinfo  {journal} {J. Phys. A: Math. Theor.}\ }\textbf {\bibinfo {volume}
  {46}},\ \bibinfo {pages} {025302} (\bibinfo {year} {2012})}\BibitemShut
  {NoStop}%
\bibitem [{\citenamefont {Fuchs}\ and\ \citenamefont
  {Peres}(1996)}]{Fuchs1996}%
  \BibitemOpen
  \bibfield  {author} {\bibinfo {author} {\bibfnamefont {C.~A.}\ \bibnamefont
  {Fuchs}}\ and\ \bibinfo {author} {\bibfnamefont {A.}~\bibnamefont {Peres}},\
  }\bibfield  {title} {\bibinfo {title} {Quantum-state disturbance versus
  information gain: Uncertainty relations for quantum information},\ }\href
  {https://doi.org/10.1103/PhysRevA.53.2038} {\bibfield  {journal} {\bibinfo
  {journal} {Phys. Rev. A}\ }\textbf {\bibinfo {volume} {53}},\ \bibinfo
  {pages} {2038} (\bibinfo {year} {1996})}\BibitemShut {NoStop}%
\bibitem [{\citenamefont {Silva}\ \emph {et~al.}(2015)\citenamefont {Silva},
  \citenamefont {Gisin}, \citenamefont {Guryanova},\ and\ \citenamefont
  {Popescu}}]{Silva2015}%
  \BibitemOpen
  \bibfield  {author} {\bibinfo {author} {\bibfnamefont {R.}~\bibnamefont
  {Silva}}, \bibinfo {author} {\bibfnamefont {N.}~\bibnamefont {Gisin}},
  \bibinfo {author} {\bibfnamefont {Y.}~\bibnamefont {Guryanova}},\ and\
  \bibinfo {author} {\bibfnamefont {S.}~\bibnamefont {Popescu}},\ }\bibfield
  {title} {\bibinfo {title} {Multiple observers can share the nonlocality of
  half of an entangled pair by using optimal weak measurements},\ }\href
  {https://doi.org/10.1103/PhysRevLett.114.250401} {\bibfield  {journal}
  {\bibinfo  {journal} {Phys. Rev. Lett.}\ }\textbf {\bibinfo {volume} {114}},\
  \bibinfo {pages} {250401} (\bibinfo {year} {2015})}\BibitemShut {NoStop}%
\bibitem [{\citenamefont {Mohan}\ \emph {et~al.}(2019)\citenamefont {Mohan},
  \citenamefont {Tavakoli},\ and\ \citenamefont {Brunner}}]{Mohan2019}%
  \BibitemOpen
  \bibfield  {author} {\bibinfo {author} {\bibfnamefont {K.}~\bibnamefont
  {Mohan}}, \bibinfo {author} {\bibfnamefont {A.}~\bibnamefont {Tavakoli}},\
  and\ \bibinfo {author} {\bibfnamefont {N.}~\bibnamefont {Brunner}},\
  }\bibfield  {title} {\bibinfo {title} {Sequential random access codes and
  self-testing of quantum measurement instruments},\ }\href
  {https://doi.org/10.1088/1367-2630/ab3773} {\bibfield  {journal} {\bibinfo
  {journal} {New J. Phys.}\ }\textbf {\bibinfo {volume} {21}},\ \bibinfo
  {pages} {083034} (\bibinfo {year} {2019})}\BibitemShut {NoStop}%
\bibitem [{\citenamefont {Renes}\ \emph {et~al.}(2004)\citenamefont {Renes},
  \citenamefont {Blume-Kohout}, \citenamefont {Scott},\ and\ \citenamefont
  {Caves}}]{Renes2004}%
  \BibitemOpen
  \bibfield  {author} {\bibinfo {author} {\bibfnamefont {J.~M.}\ \bibnamefont
  {Renes}}, \bibinfo {author} {\bibfnamefont {R.}~\bibnamefont {Blume-Kohout}},
  \bibinfo {author} {\bibfnamefont {A.~J.}\ \bibnamefont {Scott}},\ and\
  \bibinfo {author} {\bibfnamefont {C.~M.}\ \bibnamefont {Caves}},\ }\bibfield
  {title} {\bibinfo {title} {Symmetric informationally complete quantum
  measurements},\ }\href {https://doi.org/10.1063/1.1737053} {\bibfield
  {journal} {\bibinfo  {journal} {Journal of Mathematical Physics}\ }\textbf
  {\bibinfo {volume} {45}},\ \bibinfo {pages} {2171} (\bibinfo {year}
  {2004})}\BibitemShut {NoStop}%
\bibitem [{\citenamefont {Buscemi}\ and\ \citenamefont
  {Sacchi}(2006)}]{Buscemi2006}%
  \BibitemOpen
  \bibfield  {author} {\bibinfo {author} {\bibfnamefont {F.}~\bibnamefont
  {Buscemi}}\ and\ \bibinfo {author} {\bibfnamefont {M.~F.}\ \bibnamefont
  {Sacchi}},\ }\bibfield  {title} {\bibinfo {title} {Information-disturbance
  trade-off in quantum-state discrimination},\ }\href
  {https://doi.org/10.1103/PhysRevA.74.052320} {\bibfield  {journal} {\bibinfo
  {journal} {Phys. Rev. A}\ }\textbf {\bibinfo {volume} {74}},\ \bibinfo
  {pages} {052320} (\bibinfo {year} {2006})}\BibitemShut {NoStop}%
\bibitem [{\citenamefont {Cai}\ \emph {et~al.}(2025)\citenamefont {Cai},
  \citenamefont {Ren}, \citenamefont {Feng}, \citenamefont {Zhou},\ and\
  \citenamefont {Chen}}]{Cai2025}%
  \BibitemOpen
  \bibfield  {author} {\bibinfo {author} {\bibfnamefont {Z.}~\bibnamefont
  {Cai}}, \bibinfo {author} {\bibfnamefont {C.}~\bibnamefont {Ren}}, \bibinfo
  {author} {\bibfnamefont {T.}~\bibnamefont {Feng}}, \bibinfo {author}
  {\bibfnamefont {X.}~\bibnamefont {Zhou}},\ and\ \bibinfo {author}
  {\bibfnamefont {J.}~\bibnamefont {Chen}},\ }\bibfield  {title} {\bibinfo
  {title} {A review of quantum correlation sharing: The recycling of quantum
  correlations triggered by quantum measurements},\ }\href
  {https://doi.org/https://doi.org/10.1016/j.physrep.2024.10.003} {\bibfield
  {journal} {\bibinfo  {journal} {Physics Reports}\ }\textbf {\bibinfo {volume}
  {1098}},\ \bibinfo {pages} {1} (\bibinfo {year} {2025})},\ \bibinfo {note} {a
  review of quantum correlation sharing: The recycling of quantum correlations
  triggered by quantum measurements}\BibitemShut {NoStop}%
\bibitem [{\citenamefont {Bowles}\ \emph {et~al.}(2020)\citenamefont {Bowles},
  \citenamefont {Baccari},\ and\ \citenamefont {Salavrakos}}]{Bowles2020}%
  \BibitemOpen
  \bibfield  {author} {\bibinfo {author} {\bibfnamefont {J.}~\bibnamefont
  {Bowles}}, \bibinfo {author} {\bibfnamefont {F.}~\bibnamefont {Baccari}},\
  and\ \bibinfo {author} {\bibfnamefont {A.}~\bibnamefont {Salavrakos}},\
  }\bibfield  {title} {\bibinfo {title} {Bounding sets of sequential quantum
  correlations and device-independent randomness certification},\ }\href
  {https://doi.org/10.22331/q-2020-10-19-344} {\bibfield  {journal} {\bibinfo
  {journal} {{Quantum}}\ }\textbf {\bibinfo {volume} {4}},\ \bibinfo {pages}
  {344} (\bibinfo {year} {2020})}\BibitemShut {NoStop}%
\bibitem [{\citenamefont {Guerini}\ \emph {et~al.}(2017)\citenamefont
  {Guerini}, \citenamefont {Bavaresco}, \citenamefont {Terra~Cunha},\ and\
  \citenamefont {Acín}}]{Guerini2017}%
  \BibitemOpen
  \bibfield  {author} {\bibinfo {author} {\bibfnamefont {L.}~\bibnamefont
  {Guerini}}, \bibinfo {author} {\bibfnamefont {J.}~\bibnamefont {Bavaresco}},
  \bibinfo {author} {\bibfnamefont {M.}~\bibnamefont {Terra~Cunha}},\ and\
  \bibinfo {author} {\bibfnamefont {A.}~\bibnamefont {Acín}},\ }\bibfield
  {title} {\bibinfo {title} {Operational framework for quantum measurement
  simulability},\ }\href {http://dx.doi.org/10.1063/1.4994303} {\bibfield
  {journal} {\bibinfo  {journal} {Journal of Mathematical Physics}\ }\textbf
  {\bibinfo {volume} {58}} (\bibinfo {year} {2017})}\BibitemShut {NoStop}%
\bibitem [{\citenamefont {Tavakoli}\ \emph {et~al.}(2019)\citenamefont
  {Tavakoli}, \citenamefont {Rosset},\ and\ \citenamefont
  {Renou}}]{Tavakoli2019}%
  \BibitemOpen
  \bibfield  {author} {\bibinfo {author} {\bibfnamefont {A.}~\bibnamefont
  {Tavakoli}}, \bibinfo {author} {\bibfnamefont {D.}~\bibnamefont {Rosset}},\
  and\ \bibinfo {author} {\bibfnamefont {M.-O.}\ \bibnamefont {Renou}},\
  }\bibfield  {title} {\bibinfo {title} {Enabling computation of correlation
  bounds for finite-dimensional quantum systems via symmetrization},\ }\href
  {https://doi.org/10.1103/PhysRevLett.122.070501} {\bibfield  {journal}
  {\bibinfo  {journal} {Phys. Rev. Lett.}\ }\textbf {\bibinfo {volume} {122}},\
  \bibinfo {pages} {070501} (\bibinfo {year} {2019})}\BibitemShut {NoStop}%
\bibitem [{\citenamefont {Tavakoli}(2020)}]{Tavakoli2020a}%
  \BibitemOpen
  \bibfield  {author} {\bibinfo {author} {\bibfnamefont {A.}~\bibnamefont
  {Tavakoli}},\ }\bibfield  {title} {\bibinfo {title} {Semi-device-independent
  certification of independent quantum state and measurement devices},\ }\href
  {https://doi.org/10.1103/PhysRevLett.125.150503} {\bibfield  {journal}
  {\bibinfo  {journal} {Phys. Rev. Lett.}\ }\textbf {\bibinfo {volume} {125}},\
  \bibinfo {pages} {150503} (\bibinfo {year} {2020})}\BibitemShut {NoStop}%
\bibitem [{\citenamefont {Kotowski}\ and\ \citenamefont
  {Oszmaniec}(2025)}]{Kotowski2025}%
  \BibitemOpen
  \bibfield  {author} {\bibinfo {author} {\bibfnamefont {M.}~\bibnamefont
  {Kotowski}}\ and\ \bibinfo {author} {\bibfnamefont {M.}~\bibnamefont
  {Oszmaniec}},\ }\href {https://arxiv.org/abs/2501.09339} {\bibinfo {title}
  {Pretty-good simulation of all quantum measurements by projective
  measurements}} (\bibinfo {year} {2025}),\ \Eprint
  {https://arxiv.org/abs/2501.09339} {arXiv:2501.09339 [quant-ph]} \BibitemShut
  {NoStop}%
\bibitem [{\citenamefont {Descartes}(1637)}]{Descartes1637}%
  \BibitemOpen
  \bibfield  {author} {\bibinfo {author} {\bibfnamefont {R.}~\bibnamefont
  {Descartes}},\ }\href@noop {} {\bibinfo {title} {La {G}éométrie ({D}iscours
  de la méthode pour bien conduire sa raison, et chercher la vérité dans les
  sciences}} (\bibinfo {year} {1637})\BibitemShut {NoStop}%
\bibitem [{\citenamefont {Sullivan}(2017)}]{Sullivan2017}%
  \BibitemOpen
  \bibfield  {author} {\bibinfo {author} {\bibfnamefont {M.}~\bibnamefont
  {Sullivan}},\ }\href@noop {} {\emph {\bibinfo {title} {Precalculus}}}\
  (\bibinfo  {publisher} {Pearson},\ \bibinfo {year} {2017})\BibitemShut
  {NoStop}%
\end{thebibliography}%

\onecolumngrid
\appendix

\newpage

\section{ Redundancy of classical post-processing of outcomes}\label{app:post}
In this section, we show that classical post-processing in the projective instrument simulation model can be neglected without loss of generality.  

To this end, we first consider a deterministic post-processing rule. This takes the form $p(a|a',\lambda) = \delta_{a,f(a',\lambda)}$, where $f$ is some function of $a'$ and $\lambda$. In this case, Eq. (1) from the main text can be written as

\begin{align}\label{eq:mod}
	\mathcal I_a(\rho) = \sum_\lambda q_\lambda \sum_{a'} \Lambda_{a',\lambda} \left[ E_{a'\lvert \lambda}\rho E_{a'\lvert \lambda}\right] \delta_{a, f(a',\lambda)} = \sum_\lambda q_\lambda \Gamma_{a\lvert\lambda}^f,
\end{align}where we have defined $\Gamma_{a\lvert\lambda}^f =\sum_{a'} \Lambda_{a',\lambda} \left[ E_{a'\lvert \lambda}\rho E_{a'\lvert \lambda}\right] \delta_{a, f(a',\lambda)} $. We note that the set  $\{\Gamma_{a\lvert\lambda}^f\}_{a} $ defines a quantum instrument for every choice of $f$ and $\lambda$, i.e.~each $\Gamma_{a\lvert\lambda}^f$ is completely positive and the sum is trace-preserving since $\tr\left(\sum_a\Gamma_{a\lvert\lambda}^f (\rho)\right)=1$. Next, we define projective measurements $\tilde E^f_{a\lvert \lambda}\coloneqq\sum_{a'}E_{a'\lvert\lambda}\delta_{a,f(a',\lambda)}$. It can be easily seen that $\sum_a\tilde E^f_{a\lvert\lambda}=\mathds 1$, and that
\begin{equation}
	\begin{aligned}
		\tilde E^f_{a\lvert \lambda}\tilde E^f_{b\lvert \lambda} &= \sum_{a'b'}E_{a'\lvert\lambda}E_{b'|\lambda} \delta_{a,f(a',\lambda)}\delta_{b,f(b',\lambda)} = \sum_{a'b'}\delta_{a',b'}\delta_{a,f(a',\lambda)}\delta_{b,f(b',\lambda)} E_{a'|\lambda} = \sum_{a'} \delta_{a,f(a',\lambda)}\delta_{b,f(a',\lambda)}E_{a'|\lambda} = \delta_{a,b}E_{a|\lambda}.
	\end{aligned}
\end{equation}
We would like $\Gamma^f_{a|\lambda}$ to be expressible as 
\begin{align}
	\Gamma^f_{a|\lambda}(\rho) \stackrel{!}{=} \tilde \Lambda_{\lambda}\left[\tilde E^f_{a|\lambda}\rho\tilde E^f_{a|\lambda} \right] = \tilde \Lambda_{\lambda}\Big[ \sum_{l,l'} \delta_{a,f(l,\lambda)}\delta_{a,f(l',\lambda)}E_{l|\lambda}\rho E_{l'|\lambda}\Big].
\end{align}We therefore select the CPTP maps as  $\tilde \Lambda_{\lambda}[X]\coloneqq \sum_{a'}\Lambda_{a',\lambda}[E_{a'|\lambda}X E_{a'|\lambda}]$. Inserted in the above,
\begin{equation} 
	\begin{aligned}
		\tilde \Lambda_{\lambda}\left[\tilde E^f_{a|\lambda}\rho\tilde E^f_{a|\lambda} \right]=\sum_{a'}\Lambda_{a',\lambda}\Big[ E_{a'|\lambda}\sum_{l,l'} \delta_{a,f(l,\lambda)}\delta_{a,f(l',\lambda)}E_{l|\lambda}E_{l'|\lambda}E_{a'|\lambda}\Big] = \sum_{a'} \Lambda_{a',\lambda}\left[ E_{a'|\lambda}\rho E_{a'|\lambda}\right]\delta_{a,f(a',\lambda)}
	\end{aligned}
\end{equation}where in the last step, we have utilised the orthonormality of the projectors. We have obtained precisely the form in Eq.~\eqref{eq:mod}. Hence, a generic projective instrument with deterministic classical post-processing can be written as 
\begin{align}
	\mathcal I_a(\rho) =  \sum_\lambda q_\lambda \sum_{a'} \Lambda_{a',\lambda} \left[ E_{a'\lvert \lambda}\rho E_{a'\lvert \lambda}\right] \delta_{a, f(a',\lambda)}=  \sum_\lambda q_\lambda\tilde \Lambda_{\lambda}\left[\tilde E^f_{a|\lambda}\rho\tilde E^f_{a|\lambda} \right],
\end{align} where $\tilde \Lambda_{a,\lambda}$ is a CPTP channel for all $a$ and $\lambda$ and $\{\tilde E^f_{a|\lambda}\}_a$ is a projective measurement.

To complete the proof, we can directly extend the argument to stochastic post-processing rules. Any general post-processing strategy $p(a|a',\lambda)$ can be written as a convex combination over deterministic ones, $p(a|a',\lambda) = \sum_f p_f \delta_{a,f(a',\lambda)}$, where $p_f$ is the probability of selecting the deterministic post-processing function $f$. With general post-processing, Eq. (1) in the main text can therefore be written as 
\begin{align}
	\mathcal I_a(\rho) = \sum_\lambda q_\lambda\sum_f p_f \sum_{a'} \Lambda_{a',\lambda}\left[E_{a'|\lambda}\rho E_{a'|\lambda} \right]\delta_{a,f(a',\lambda)} = \sum_\lambda q_\lambda \sum_f p_f \ \tilde \Lambda_{\lambda}\left[\tilde E^f_{a|\lambda}\rho\tilde E^f_{a|\lambda} \right].
\end{align}
We can now simply re-define $\lambda \rightarrow (\lambda,f)$. Its new  distribution becomes $q_\lambda\rightarrow q_\lambda p_f$. We have then obtained a projective instrument without classical post-processing. Moreover, the quantum post-processing channel no longer needs to depend on the classical outcome, which can be absorbed into the channel itself.

\section{Proof of Theorem 1}\label{app:thm1}

\textbf{Theorem 1} (Instrument simulation).
Consider any $N$-outcome projective instrument and let  $\{\eta_a\}_{a=1}^N$ be its  Choi representation. There exists a decomposition 
\begin{align}\label{choiA}\nonumber
	&\qquad  \eta_a= \sum_{\vec{r}} \sigma_{a|\vec{r}}, \qquad \text{where} \qquad \sigma_{a|\vec{r}}\in\mathcal{L}_+(\mathcal{H}_{A'}\otimes\mathcal{H}_A)\\
	&\tr(\sigma_{a|\vec{r}})=q_{\vec{r}}\frac{r_a}{d},\qquad  \sum_{a}\tr_{A'}(\sigma_{a|\vec{r}})=q_{\vec{r}}\frac{\openone}{d}, \qquad  \text{SN}(\sigma_{a|\vec{r}})\leq r_a.
\end{align}
where $\vec{r}$ runs over all rank-vectors and $\text{SN}$ denotes the Schmidt number.

\textit{Proof.-- }Consider a $N$-outcome projective instrument $\boldsymbol{\mathcal I}$ operating from a Hilbert space of dimension $d=\text{dim}(\mathcal{H}_A)$. For outcome $a$, the post-measurement state is by definition given by
\begin{equation}\label{PIapp}
	\mathcal{I}_a(\rho)=\sum_\lambda q_\lambda  \Lambda_{\lambda}\left[E_{a|\lambda}\rho E_{a|\lambda}\right],
\end{equation}
where $\{E_{a|\lambda}\}_a$ are projective measurements, $\{\Lambda_{\lambda}\}$ are CPTP maps and $\{q_\lambda\}$ is a probability distribution. 
As in the main text, we associate every $N$-outcome projective measurement, $\mathbf{E}=\{E_a\}_{a=1}^N$, with a rank-vector $\vec{r}=(r_1,\ldots,r_N)$, where $r_a=\rank(E_a)$. Thus, every $N$-tuple of non-negative integers such that $\sum_a r_a=d$ is a valid rank-vector. Therefore, we write $\lambda=(\chi,\vec{r})$, where $\chi$ is a random variable for selecting projective measurements with rank-vector $\vec{r}$. We can thus write Eq. \eqref{PIapp} equivalently as 
\begin{equation}
	\begin{aligned}
		\mathcal{I}_a(\rho) =\sum_{\chi,\vec r} q_{\chi,\vec r}  \Lambda_{\chi,\vec r}\left[E_{a|\chi,\vec r}\,\rho E_{a|\chi,\vec r}\right]
		\equiv \sum_{\chi,\vec r} q_{\chi,\vec r} \,\mathcal S_{a\lvert \chi,\vec r}[\rho], 
	\end{aligned}
\end{equation}
where $\mathcal S$ is a completely positive trace-non-increasing map. We now use channel-state duality to represent these maps as bipartite operators. We associate the Choi operator $\eta_a$ to the instrument $\mathcal I_a$ and the Choi operator $\tilde \sigma_{a\lvert \chi,\vec r}$ to $\mathcal S_{a\lvert \chi,\vec r}$. Therefore, we have 
\begin{align}
	\eta_a = \sum_{\chi,\vec r} q_{\chi,\vec r} \left(\mathcal S_{a\lvert \chi,\vec r}\otimes \mathds 1\right)[\phi^+] = \sum_{\chi,\vec r} q_{\chi,\vec r}\, \tilde  \sigma_{a\lvert\chi,\vec r} =\sum_{\vec{r}}\sigma_{a|\vec{r}},
\end{align}
where we have defined $\sigma _{a\lvert \vec r } \coloneqq \sum_{\chi} q_{\chi,\vec r}\, \tilde \sigma_{a\lvert \chi,\vec r}$. Note that $\sigma_{a\lvert \vec r} \in \mathcal{L}_+(\mathcal{H}_{A'}\otimes\mathcal{H}_A)$ because of the complete positivity of the involved maps.  This corresponds to the first constraint in \eqref{choiA}.

Next, we have 
\begin{equation}
	\begin{aligned}
		\tr(\sigma_{a\lvert \vec r}) &= \sum_\chi q_{\chi,\vec r} \tr\left(\tilde \sigma_{a\lvert \chi,\vec r}\right) = \sum_\chi q_{\chi,\vec r}\tr\left(\mathcal S_{a\lvert \chi,\vec r}\otimes \mathds 1[\phi^+] \right) 
		=\frac{1}{d} \sum_\chi q_{\chi,\vec r}\sum_{i,j}\tr\left( \Lambda_{\chi,\vec r}\left[E_{a|\chi,\vec r}\,\ketbra{i}{j} E_{a|\chi,\vec r}\right] \otimes \ketbra{i}{j}\right) \\ 
		& = \frac{1}{d} \sum_\chi q_{\chi,\vec r}\tr\bigg(\Lambda_{\chi,\vec r}\big[E_{a|\chi,\vec r}\,\bigg(\sum_{i}\ketbra{i}{i}\bigg) E_{a|\chi,\vec r}\big]\bigg) = \frac{1}{d} \sum_\chi q_{\chi,\vec r}\tr\left(\Lambda_{\chi,\vec r}\big[E_{a|\chi,\vec r}\big]\right)\\
		&=\frac{1}{d} \sum_\chi q_{\chi,\vec r}\tr\left(E_{a|\chi,\vec r}\right)=q_{\vec{r}}\frac{r_a}{d},
	\end{aligned}
\end{equation}
where in the third-last step  we have first used completeness ($\openone=\sum_i \ketbra{i}{i}$) and then that $\{E_{a|\chi,\vec{r}}\}_a$ is a projective measurement. In the penultimate step we have used that  $\Lambda_{\chi,\vec r}$ is trace-preserving. In the last step, we have used that $\rank(E_{a|\chi,\vec{r}})=\tr(E_{a|\chi,\vec{r}})=r_a$ and we define the probability distribution $q_{\vec r}\coloneqq \sum_\chi q_{\chi,\vec r }$. This corresponds to the first constraint on the second row of \eqref{choiA}.

Next, we have 
\begin{equation}
	\begin{aligned}
		\sum_{a}\tr_{A'}(\sigma_{a|\vec{r}}) &= \sum_{a,\chi} q_{\chi,\vec r} \tr_{A'}\left(\tilde \sigma_{a\lvert \chi,\vec r}\right)  =  \frac{1}{d} \sum_{\chi} q_{\chi,\vec r}\sum_{i,j}\sum_a\tr\bigg( \Lambda_{\chi,\vec r}\left[E_{a|\chi,\vec r}\,\ketbra{i}{j} E_{a|\chi,\vec r}\right] \bigg)\otimes \ketbra{i}{j}.
	\end{aligned}
\end{equation} 
Using first the trace-preservation of $\Lambda_{\chi,\vec{r}}$, then the projectivity of $E_{a|\chi,\vec{r}}$ and finally the completeness of the measurement:
\begin{equation}
	\sum_a\tr\left( \Lambda_{\chi,\vec r}\left[E_{a|\chi,\vec r}\,\ketbra{i}{j} E_{a|\chi,\vec r}\right] \right)=\sum_a\tr\left( E_{a|\chi,\vec r}\,\ketbra{i}{j} E_{a|\chi,\vec r} \right)=\tr\left( \ketbra{i}{j} \sum_aE_{a|\chi,\vec r} \right)=\delta_{i,j}.
\end{equation}
Hence,
\begin{equation}
	\sum_{a}\tr_{A'}(\sigma_{a|\vec{r}})=\frac{1}{d}\sum_{\chi} q_{\chi,\vec{r}} \sum_{i,j}\delta_{i,j}\ketbra{i}{j}=q_{\vec{r}}\frac{\openone}{d}.
\end{equation}
This is the second constraint in the second row of \eqref{choiA}.

Lastly, the local projection of the maximally entangled state $\phi^+$ onto $E_{a\lvert \chi,\vec r}$ confines it to a $ r_a$-dimensional subspace, i.e.,the sub-normalised state $(E_{a\lvert \chi,\vec r}\otimes \mathds 1)\phi^+(E_{a\lvert \chi,\vec r}\otimes \mathds 1)$ lives on a Hilbert space isomorphic to $\mathcal{L}_+(\mathbb{C}^{r_a}\otimes \mathbb{C}^d)$. Such a state trivially has Schmidt number no larger than its smallest local dimension, namely $r_a$. Next, we are allowed to stochastically apply CPTP maps, $\Lambda_{\chi,\vec r}$,  on one of the systems. However, these are local operations without post-selection; such operations cannot increase the Schmidt number \cite{Terhal2000}. Hence we conclude that $\text{SN}(\sigma_{a\lvert\vec r})\leq  r_a$. This is the third constraint on the second row of \eqref{choiA}.\qed

\section{Proof of Theorem 2}\label{app:thm2}

\textbf{Theorem 2} (Qubit instrument simulation)
\textit{For instruments with qubit input, Theorem 1 is both necessary and sufficient.} \\

\textit{Proof.--} Consider and $N$-outcome qubit instrument with Choi operators $\eta_a$. Theorem 1 is necessary for the instrument to be a PI. Here, we prove that it is also sufficient. 

Suppose there exist $\sigma_{a\lvert \vec r}$ such that the conditions \eqref{choiA} are satisfied. Qubit projective measurements only have two qualitatively different types of rank-vectors. These are $\vec{s}=(2,0,\ldots,0)$ and $\vec{t}=(1,1,0,\ldots,0)$, up to permutations. We denote the permutation of $\vec s$ with "2" as the $i$-th element as $\vec s_i$. Likewise, we denote the permutation of $\vec t$ with "1" as the $j$-th  and  $k$-th elements ($j\neq k$) as $\vec t_{jk}$. Furthermore, for rank vector $\vec s_i$, outcome $a=i$ corresponds to the rank-2 projector $E_{i\lvert \vec s_i}=\mathds 1$, while other outcomes correspond to null matrices. On the other hand, for the rank vector $\vec t_{jk}$, outcomes $a=j,k$ are both relevant and correspond to orthogonal rank-1 projectors $E_{a\lvert \chi,\vec t_{jk}} = \varphi_{a\lvert\chi,\vec t_{jk}}$ (such that $\varphi_{j\lvert\chi,\vec t_{jk}} = \mathds 1-\varphi_{k\lvert\chi,\vec t_{jk}} $), while other outcomes correspond to null matrices. Note that projective simulability of the instrument would imply
\begin{equation}
	\begin{aligned}
		\eta_a &= \sum_{i=1}^{N} q_{\vec s_i}(\Lambda_{\vec s_i}\otimes \mathds 1) [E_{a\lvert \vec s_i}\otimes \openone\,\phi^+E_{a\lvert \vec s_i}\otimes \openone] +\sum_{j<k} \sum_{\chi} q_{\chi,\vec t_{jk}} (\Lambda_{\chi,\vec t_{jk}}\otimes \mathds 1)[E_{a\lvert \chi,\vec t_{jk}}\otimes \openone\,\phi^+E_{a\lvert \chi,\vec t_{jk}}\otimes \openone]  \\
		&=  q_{\vec s_a}(\Lambda_{\vec s_a}\otimes \mathds 1) [\phi^+] + \sum_{j<k}\sum_{\chi} q_{\chi,\vec t_{jk}}\,\frac{1}{2}\Lambda_{\chi,\vec t_{jk}}[ \varphi_{a\lvert\chi,\vec t_{jk}}]\otimes\varphi^T_{a\lvert\chi,\vec t_{jk}},
	\end{aligned}
\end{equation}where in  we have used the specific form of the projectors and the identity $(E\otimes\openone)\phi^+ (E\otimes\openone) = (E\otimes E^T)/2$ for any projector $E$. The first term in the above sum (second row) corresponds to a generic form of a sub-normalised Choi state, which corresponds to the constraints in \eqref{choiA} for $\sigma_{a|\vec{s}_a}$; note that  $\text{SN}\leq 2$ trivially holds for all $\sigma_{a|\vec{r}}\in\mathcal{L}_+\left(\mathcal{H}_{A'}\otimes \mathbb{C}^2\right)$. Consider now the second term. For PIs, $\Lambda_{\chi,\vec{t}_{jk}}$ is an arbitrary CPTP map. We can therefore, whenever $a\in\{j,k\}$, generate any desired state $\xi_{a,\chi,\vec t_{jk}}=\Lambda_{\chi,\vec{t}_{jk}}[\varphi_{a\lvert\chi,\vec t_{jk}}]$. The second term then becomes
\begin{align}
	\sum_{j<k} \sum_{\chi} q_{\chi,\vec t_{jk}} \frac{1}{2}\xi_{a,\chi,\vec t_{jk}}\otimes\varphi^T_{a\lvert\chi,\vec t_{jk}}=\sum_{j<k}
	\begin{cases}
		T_{a,jk}& a\in\{j,k\}\\
		0 & \text{otherwise}
	\end{cases},
\end{align}
where $T_{a,jk}=\frac{1}{2}\sum_{\chi} q_{\chi,\vec t_{jk}}\xi_{a,\chi,\vec t_{jk}}\otimes\varphi^T_{a\lvert\chi,\vec t_{jk}}$ is a separable operator with  $\tr(T_{a,jk})=\frac{1}{2}\sum_\chi q_{\chi,\vec t_{jk}}=\frac{q_{\vec t_{jk}}}{2}$. Thus, each $T_{a,jk}$ becomes becomes the generic form of a separable (SN $=1$)  trace-half (up to sub-normalisation $q_{\vec t_{jk}}$) operator. This  corresponds precisely to the conditions \eqref{choiA} for $\sigma_{a|\vec{t}_{jk}}$. Hence, the conditions \eqref{choiA} imply projective simulability for qubit instruments.
\qed

\section{Simulation of noisy unsharp Pauli measurement process}\label{app:weakz}
Consider the L\"uders instrument corresponding to an unsharp Pauli observable. It has Kraus operators
\begin{equation}
	K_{a}=\sqrt{\frac{1-(-1)^a \gamma}{2}}\ketbra{0}{0}+\sqrt{\frac{1+(-1)^a \gamma}{2}}\ketbra{1}{1},
\end{equation}
for outcomes $a\in\{1,2\}$ and sharpness parameter $\gamma\in[0,1]$. The Choi representation of this instrument reads $\eta_a=(K_a\otimes \openone)\phi^+ (K_a^\dagger\otimes\openone)$. We consider now the instrument that corresponds to a mixture of the above extremal instrument with a noise instrument $\{\eta_a^\text{noise}\}_a$. The mixture becomes
\begin{equation}\label{noiseX}
	\eta_a^v=v\eta_a+(1-v)\eta_a^\text{noise},
\end{equation} 
for some visibility $v\in[0,1]$. We are interested in determining the critical visibility for PI-simulation, i.e.~we seek the largest $v$ such that $\{\eta_a^v\}\in\mathcal{P}$. For PIs, following the necessary and sufficient SDP characterisation is obtained from Theorem 1 and Eq.~\eqref{choiA},
\begin{align}\nonumber\label{critvisi}
	\max &\quad v\\\nonumber
	\quad& \text{such that} \quad  \eta_a^v=\sigma_{a|(1,1)}+\sigma_{a|(2,0)}+\sigma_{a|(0,2)}, \\ \nonumber
	\quad & \tr(\sigma_{1|(1,1)})=\tr(\sigma_{2|(1,1)})=\frac{q_{(1,1)}}{2}, \qquad \tr(\sigma_{1|(2,0)})=q_{(2,0)}, \qquad  \tr(\sigma_{2|(0,2)})=q_{(0,2)}\\\nonumber
	\quad & \sum_a \tr_{A'}(\sigma_{a|(1,1)})=q_{(1,1)}\frac{\openone}{2}, \qquad\tr_{A'}(\sigma_{1|(2,0)})=q_{(2,0)}\frac{\openone}{2}, \qquad  \tr_{A'}(\sigma_{2|(0,2)})=q_{(0,2)}\frac{\openone}{2}\\\nonumber
	\quad & \sigma_{1|(1,1)}^{T_{A}}\succeq 0, \qquad \sigma_{2|(1,1)}^{T_{A}}\succeq 0, \qquad \sigma_{2|(2,0)}=\sigma_{1|(0,2)}=0,\\
	\quad & \sigma_{1|(1,1)},\sigma_{2|(1,1)},\sigma_{1|(2,0)},\sigma_{2|(0,2)}\in\mathcal{L}_+\left(\mathbb{C}_{A'}^2\otimes \mathbb{C}_{A}^2\right).
\end{align}
We will study this problem for three prominent noise models, namely
\begin{align}
	& \text{dephasing noise:} && \eta_a^\text{noise}=\frac{1}{4}\left(\ketbra{00}+\ketbra{11}\right),\\
	& \text{worst-case noise:} && \{\eta_a^\text{noise}\}_a \text{ is an arbitrary instrument}\\
	& \text{white noise:} && \eta_a^\text{noise}=\frac{1}{8}\openone.
\end{align}

Below, we address the three noise models one by one and analytically derive the critical visibility for PI-simulation. All results are optimal, i.e.~they are the exact solution of \eqref{critvisi}, since they have been matched up to solver precision via the SDP characterisation \eqref{critvisi}.

\subsection{Dephasing noise}
Select the variables as $\sigma_{a|(1,1)}=\frac{q_{(1,1)}}{2}\ketbra{aa}$ and $\sigma_{1|(2,0)}=\sigma_{2|(0,2)}=\frac{1-q_{(1,1)}}{2} \phi^+$. The last four rows of constraints in Eq.~\eqref{critvisi} are satisfied by construction. Solving for the first constraint, namely the simulation equation  in \eqref{noiseX}, we obtain $q_{(1,1)}=\gamma v_{\text{deph}}$ and
\begin{equation}
	v_\text{deph}=\frac{1}{\gamma+\sqrt{1-\gamma^2}}.
\end{equation}

\subsection{Worst-case noise}

We use the simulation strategy and $\sigma_{a\lvert(1,1)}=\frac{q_{(1,1)}}{2}\ketbra{aa}{aa}$ and $\sigma_{1\lvert (2,0)} =\sigma_{2\lvert (0,2)}=\frac{1-q_{(1,1)}}{2}\phi^+$. We must also choose the noise instrument optimally. We select
\begin{align}
	\eta_1^{\text{noise}} = \begin{pmatrix}
		x_1 &0 &0 & -\sqrt{x_1x_2}\\ 0&0&0&0\\0&0&0&0\\ -\sqrt{x_1x_2} &0&0&x_2
	\end{pmatrix}, \quad \eta_2^{\text{noise}} = \begin{pmatrix}
		x_2 &0 &0 & -\sqrt{x_1x_2}\\ 0&0&0&0\\0&0&0&0\\ -\sqrt{x_1x_2} &0&0&x_1
	\end{pmatrix}.
\end{align}Using Eq. \eqref{critvisi} and the constraints \eqref{noiseX}, we obtain
\begin{equation}
	\begin{aligned}
		q_{(1,1)} = \frac{1}{2}\left(\left(\gamma -\sqrt{1-\gamma ^2}\right) v+1\right) ,\quad x_1 = \frac{\left(3-\left(\sqrt{1-\gamma ^2}+\gamma +2\right) v\right),}{8(1-v)} \quad  x_2 = \frac{\left(\left(\sqrt{1-\gamma ^2}+\gamma -2\right) v+1\right)}{8(1-v)},
	\end{aligned}
\end{equation}and the critical visibility
\begin{align}
	v_{\text{worst}} = \frac{1}{-\sqrt{1-\gamma ^2}+2 \sqrt{(\gamma -1) \left(\sqrt{1-\gamma ^2}-1\right)}-\gamma +2}.
\end{align}

\subsection{White noise}
Next, we consider the critical visibility for PI-simulation under white noise. Consider the following Choi operators,
\begin{align}
	&\sigma_{1\lvert(1,1)} = 	\begin{pmatrix}
		y_1 & 0&0&\sqrt{y_2 y_3}\\
		0& y_2 & 0&0\\
		0&0&y_3&0\\
		\sqrt{y_2 y_3}&0&0&y_4
	\end{pmatrix},
	\quad
	&&\sigma_{2\lvert(1,1)}	=\begin{pmatrix}
		y_4 & 0&0&\sqrt{y_2 y_3}\\
		0& y_3 & 0&0\vspace{5mm}\\
		0&0&y_2&0\\
		\sqrt{y_2 y_3}&0&0&y_1
	\end{pmatrix},\\
	&\sigma_{1\lvert(2,0)} = 	\begin{pmatrix}
		z_1 & 0&0&\sqrt{z_1 z_3}\\
		0& z_2 & 0&0\\
		0&0&0&0\\
		\sqrt{z_1 z_3}&0&0&z_3
	\end{pmatrix},
	\quad 
	&&\sigma_{2\lvert(0,2)} = 	\begin{pmatrix}
		z_3 & 0&0&\sqrt{z_1 z_3}\\
		0& 0 & 0&0\\
		0&0&z_2&0\\
		\sqrt{z_1 z_3}&0&0&z_1
	\end{pmatrix},
\end{align}
where $y_i,z_i\geq 0$. The operators  $\sigma_{a\lvert (1,1)}$ are separable by construction. The operators  $\sigma_{1\lvert (2,0)}$ and $\sigma_{2\lvert (0,2)}$ are entangled by construction and their form is obtained by ensuring that they are  positive semidefinite. Imposing the appropriate sub-system reduction constraints from \eqref{critvisi} and solving the first constraint in \eqref{critvisi} gives a complete solution in terms of all the parameters $\{\{y_i,z_i\}_i,v\}$. Almost all of them admit a closed form, except two. The critical visibility $v_{\text{white}}$ is the largest solution of the following eighth-degree equation,

\begin{equation}
	\begin{aligned}
		&\left(16384 \gamma ^8-12288 \gamma ^7-17408 \gamma ^6+18496 \gamma ^5+448 \gamma ^4-4624 \gamma ^3+1768 \gamma ^2-375 \gamma \right) v^8\\&+\left(-4096 \gamma ^7-14336 \gamma ^6-1728 \gamma ^5+20160 \gamma ^4-1392 \gamma ^3-9280 \gamma ^2+3065 \gamma -625\right) v^7\\& +\left(-1024 \gamma ^6+10432 \gamma ^5+832 \gamma ^4-992 \gamma ^3+3176 \gamma ^2-4787 \gamma +1575\right) v^6 \\ &+\left(5568 \gamma ^5-3520 \gamma ^4-8480 \gamma ^3+3328 \gamma ^2+1645 \gamma -1061\right) v^5+\left(-1536 \gamma ^4-1552 \gamma ^3+888 \gamma ^2+795 \gamma -29\right) v^4\\ &+\left(656 \gamma ^3+64 \gamma ^2-437 \gamma +77\right) v^3 +\left(56 \gamma ^2+79 \gamma +53\right) v^2+(15 \gamma +9) v+1=0.
	\end{aligned}
\end{equation}
$y_2$ is the real part of either of the roots of the following equation, i.e., $y_2 = \text{Re}(x)$,
\begin{multline}
	\sqrt{64 x^2-v \left(32 \sqrt{2} \sqrt{x\left(\gamma ^2-1\right) (v-1)}+16 \gamma ^2 v-17 v+2\right)-16 x (v-1)+1}\\
	+\sqrt{64 x^2+16x (4 \gamma  v+1-v)	+(4 \gamma  v+v-1)^2} = 2(1+v) 
\end{multline}

The other parameters are 
\begin{align}
	y_1 = \frac{1}{16} \left(\xi+8 y_2+4 \gamma  v+3 v+1\right), \,\, y_4 = \frac{1}{16} \left(\xi-8 y_2-4 \gamma  v+v+3\right),\,\,	z_1 = \frac{1}{16} \left(\xi-8 y_2-v+1\right),
\end{align} $	y_3 = (1 - v)/8$, $z_2= (1 - 8 y_2 - v)/8$ and  $z_3=1/2-(y_1+y_2+y_3+y_4+z_1+z_2)$, where we have defined
\begin{align}
	\xi \coloneqq \sqrt{64 y_2^2-v \left(32 \sqrt{2} \sqrt{y_2\left(1-\gamma ^2\right)\left( 1-v\right)} +16\gamma^2v  - 17v+2\right)-16y_2 (v-1)+1}.
\end{align}

%\newpage
\section{Information-disturbance tests}\label{AppHemisphere}

\subsection{Computing general witnesses}
Our SDP characterisation allows us to compute arbitrary witnesses of instrument non-projectivity. The idea is as presented in the main text. The experimenter prepares a set of states $\{\psi_x\}$, passes them through a given instrument, $\boldsymbol{\mathcal{I}}$, to be tested for non-projectivity.  The experimenter then measures the quantum output with some POVMs $\{N_{b|y}\}_b$ where $y$ indexes the choice of measurement. The statistics of this experiment are given by the conditional probabilities $p(a,b|x,y)=\tr(\mathcal{I}_a(\psi_x)N_{b|y})$. To detect  non-projective behaviour, consider a linear witness
\begin{equation}
	W\equiv \sum_{a,b,x,y} c_{abxy} p(a,b|x,y)\leq \beta,
\end{equation}
where $c_{abxy}$ are real coefficients and $\beta$ is a bound satisfied by all PIs. In the Choi representation, $W$ can be expressed as
\begin{align}
	W=\sum_{a,b,x,y}c_{abxy}\,d\tr\left( \left( N_{b\lvert y}\otimes\psi_x^T\right)\eta_a\right),
\end{align}
where $\eta_a$ is the Choi operator corresponding to $\mathcal I_a$ and $d=\text{dim}(\mathcal H_A)$. For a PI, $W$ becomes a convex combination over its deterministic values, associated with a specific choice of the classical variable $\lambda=(\chi,\vec{r})$. Therefore, due to this linearity, the optimal witness value corresponds to a specific choice of $(\chi,\vec{r})$. Therefore, let us denote by $\mathcal{P}_{\vec{r}}$ the set of PIs in which only projective measurements with rank-vector $\vec{r}$ are employed. We need only to bound $W$ for each rank-vector individually and select the largest value. Thus, we are interested in computing  
\begin{equation}
	\beta_{\vec{r}}=\max_{\boldsymbol{\mathcal{I}}\in\mathcal{P}_{\vec{r}}} W
\end{equation}
and then set the bound to
\begin{equation}
	\beta=\max_{\vec{r}} \beta_{\vec{r}}.
\end{equation}

Using the SDP relaxation of $\mathcal{P}$  we can efficiently compute upper bounds $\beta_{\vec{r}}^\uparrow$ bounds on each $\beta_{\vec{r}}$. For PIs, the program becomes
\begin{equation}\label{prog}
	\begin{aligned}
		\beta^\uparrow_{\vec{r}} =\max_{\{\sigma_a\}}\,\, &\sum_{a,b,x,y}c_{abxy}\,d\tr\left( \left( N_{b\lvert y}\otimes\psi_x^T\right)\sigma_a\right)\\
		&\text{such that}\quad \sigma_{a}\in\mathcal{L}_+(\mathcal{H}_{A'}\otimes\mathcal{H}_A), \quad \tr(\sigma_{a})=\frac{r_a}{d},\quad  \sum_{a}\tr_{A'}(\sigma_a)=\frac{\openone}{d}, \quad  (\Theta_{r_a}\otimes\openone)[\sigma_{a}]\succeq 0
	\end{aligned}
\end{equation}
In the SDP, we have used the generalised reduction map $\Theta_s(X)=\tr(X)\openone-\frac{1}{s}X$ to relax the Schmidt number condition in Theorem 1. In the case of qubits, only  $r_a=1$ implies a non-trivial Schmidt number. We then replace the final constraint in \eqref{prog}  with positive partial transpose. This ensures that $\beta^\uparrow_{\vec{r}}=\beta_{\vec{r}}$ for qubits.

\subsection{Bloch-hemisphere discrimination }

Consider that we randomly and uniformly draw a pure qubit state $\ket{\psi}$. To determine whether it belongs to the northern (N) or the southern (S) Bloch-hemisphere, we utilise an arbitrary two-outcome qubit instrument, $\{\mathcal{I}_a\}$ with outcomes $a\in\{\text{N,S}\}$. The success probability is given by
\begin{equation}
	\begin{aligned}
		p_{\text{win}} = \frac{1}{2} \int_{\text N} d\psi \tr\left(\mathcal{I}_\text{N}\left(\psi \right)\right) + \frac{1}{2} \int_{\text S} d\psi \tr\left(\mathcal{I}_\text{S}\left(\psi \right)\right) 
	\end{aligned} 
\end{equation}
In the Choi representation, $\mathcal{I}_a(\psi) = 2\tr_A\left(\openone_{A'}\otimes \psi_A^T  \eta_a \right)$ for a bipartite operator $\eta_a\in\mathcal{L}_+(\mathbb{C}^2_{A'}\otimes\mathbb{C}^2_A)$. We then have,
\begin{align}
	p_{\text{win}} = \int_{\text{N}} d\psi \tr\left( \psi^T\eta_\text{N}^{A}\right) + \int_{\text S} d\psi \tr\left( \psi^T\left(  \frac{\mathds1}{2}-\eta_\text{N}^{A}\right)\right), 
\end{align}
where we have further utilised $\tr_{A'}\eta_N+\tr_{A'}\eta_S = \mathds 1/2$.  By noting that the the two hemispheres are connected by a reflection along the Z axis, we can write, 
\begin{align}
	p_{\text{win}} = \int_{\text S} d\psi \frac{1}{2}+ \int_{\text N} d\psi  \tr\left(\left( \psi^T - \left(\sigma_x\psi^T \sigma_x\right)^*\right)\eta_N^A \right). 
\end{align}We now explicitly evaluate the following integral over the Haar measure, $d\psi = \sin\theta\, d\theta \,d\phi/2\pi$,%\textcolor{blue}{This seemed to be the main error in the end, needs to be $2\pi$, not 4.},
\begin{equation}
	\begin{aligned}
		\int_{\text{N}}d\psi \psi^T = \frac{1}{2\pi}\int_0^{2\pi}d\phi  \int_{0}^{\pi/2}d\theta \sin\theta \ketbra{\psi(\theta)}{\psi(\theta)}  = \frac{1}{2}\ketbra{0}{0} + \frac{1}{2}\frac{\mathds 1}{2} ,
	\end{aligned}
\end{equation}with $\ket{\psi(\theta)}\coloneqq \cos(\theta/2)\ket{0} + e^{-i\phi}\sin(\theta/2)\ket{1}$. Accordingly, $(\sigma_x \int_{\text{N}}d\psi \psi^T \sigma_x )^*= \frac{1}{2}\ketbra{1}{1} + \frac{1}{2}\frac{\mathds 1}{2}$. We finally have 
\begin{equation}\label{eq:pwin}
	\begin{aligned}
		p_{\text{win}} = \frac{1}{2} + \frac{1}{2} \tr\left(\eta_N^A\sigma_z \right)
	\end{aligned}
\end{equation}The fidelity is given by 
\begin{equation}\label{eq:F}
	\begin{aligned}
		F = \int  d\psi \bra{\psi} \sum_a \mathcal I_a(\psi) \ket{\psi} = 2\int d\psi \tr\left( \psi\otimes\psi^T\eta\right) =\tr\left(\eta \left(\frac{2}{3}\Phi^+ +\frac{2}{3}\frac{\mathds 1}{2} \right) \right)  = \frac{1}{3} + \frac{2}{3} \tr\left( \Phi^+\eta\right),
	\end{aligned}
\end{equation}where $\eta=\eta_N+\eta_S$ and we have used $\sum_a\mathcal{I}_a(\psi) = 2\tr_A\left(\openone_{A'}\otimes \psi_A^T  \eta \right)$. First, we obtain the optimal-information-disturbance tradeoff for projective instruments. Taking $\sigma_{1\lvert(1,1)}=\alpha\ketbra{00}{00}$, $\sigma_{2\lvert(1,1)}=\alpha\ketbra{11}{11}$ and $\sigma_{1\lvert(2,0)}=\sigma_{2\lvert(0,2)} = (1/2-\alpha) \phi^+$. Then, $\eta_1 = \alpha\ketbra{00}{00}+(1/2-\alpha) \phi^+$ and $\eta_2 = \alpha\ketbra{11}{11}+(1/2-\alpha) \phi^+$. Note that this instrument is projective. By construction $\tr_2\sum_a \sigma_{a|\vec{r}}\propto \frac{\openone}{2}$ and $\tr_1\sum_a \sigma_{a|\vec{r}}\propto \frac{\openone}{2}$ and that normalisation holds. From Eqs. \eqref{eq:pwin} and \eqref{eq:F}, we further have $F_{\text{sim}}=1-2\alpha/3$ and $p_{\text{win}}=(1+\alpha)/2$, from which we finally obtain 
\begin{align}
	F_{\text{sim}} = \frac{5-4p_{\text{win}}}{3}.
\end{align}

Next, we obtain the optimal information-tradeoff for arbitrary quantum instruments. Taking the Choi state of the unsharp Pauli-Z instrument, we find 
\begin{align}
	F_Q = \frac{1}{3} + \frac{1+\sqrt{1-\gamma^2}}{3} = \frac{2}{3} + \frac{\sqrt{1-\gamma^2}}{3}
\end{align} and 
\begin{align}
	p_{\text{win} } = \frac{1}{2} + \frac{\gamma}{4},
\end{align}or stated in a different form,
\begin{align}
	F_Q=\frac{2}{3}+\frac{1}{3}\sqrt{16p_\text{win}(1-p_\text{win})-3}.
\end{align}
This corresponds precisely to the red curve in Fig. 3 obtained with the SDP.

\section{Data for sequential CHSH violations with local randomness and PIs}
\label{app:SeqBell}
Alice and Bob have binary inputs $x,y\in\{0,1\}$ and binary outputs $a,b\in\{0,1\}$ respectively. They share the two-qubit state  $\Psi$ on which they perform a CHSH test. Bob's input is associated with an instrument with Kraus operators $\{K_{b|y}\}_b$. The quantum output of the instrument is a qubit that is relayed to Charlie. The state then shared between Alice and Charlie becomes
\begin{equation}
	\Psi_\text{post}=\frac{1}{2}\sum_{b,y} (\openone\otimes K_{b|y})\Psi(\openone\otimes K_{b|y}^\dagger).
\end{equation}
In Choi representation this becomes
\begin{equation}
	\Psi_\text{post}=\frac{1}{2}\sum_{b,y}\tr_{B}\left[\left(\openone_{B'}\otimes \Psi_{AB}^{T_B}\right)\left(\openone_A\otimes \eta_{b|y}^{B'B}\right)\right],
\end{equation}
where $\{\eta_{b|y}^{B'B}\}$ are the Choi operators of Bob's $y$'th instrument. Let Alice and Charlie perform sharp measurements corresponding to observables $\{A_x\}$ and $\{C_z\}$ respectively, for $z=0,1$ being Charlie's input. The CHSH parameter between Alice and Bob becomes
\begin{align}
	\mathcal{S}_\text{AB}=2\sum_{x,y}(-1)^{xy}\tr\left(\Psi_{AB} \left(A_x\otimes (\eta_{0|y}^{A'}-\eta_{1|y}^{A'})\right) \right),
\end{align}
Similarly, the CHSH parameter between Alice and Charlie becomes
\begin{equation}
	\mathcal{S}_\text{AC}=\sum_{x,z}(-1)^{xz} \expect{A_x,C_z}_{\Psi_\text{post}}.
\end{equation}
For given values of $\{A_x,C_z,\Psi_\text{AB},\mathcal{S}_{\text{AB}}\}$, the optimal $\mathcal{S}_\text{AC}$ can be computed as an SDP.

In order to explore the trade-off between $\mathcal{S}_\text{AB}$ and $\mathcal{S}_\text{AC}$, we use an alternating convex search procedure - also known as a seesaw. Firstly, using Theorems 1 and 2, we constrain our SDP over Bob's instrument so that it characterises PIs acting on qubits.  Secondly, instead of optimising over both $\Psi_\text{AB}$ and $A_x$, it is sufficient to consider only the assemblage prepared remotely by Alice for Bob. This is given by $\tau_{a|x}=\tr_A\left(\frac{\mathds{1}+(-1)^a A_x}{2}\otimes\openone \rho_\text{AB}\right)$. By the GHJW theorem, every assemblage satisfying the basic properties $\tau_{a|x}\succeq 0$, $\sum_a \tau_{a|x}=\tau$ such that $\tr(\tau)=1$ can be realised by some choice of state $\Psi_\text{AB}$ and local observables $A_x$. Note that these are semidefinite constraints that can be addressed by SDP. Putting this together, the seesaw routine for optimising the trade-off becomes

\begin{enumerate}
	\item Select a random assemblage $\{\tau_{a|x}\}$ and a random measurement $\{C_z\}$.
	\item Evaluate the SDP over PIs that optimises $\mathcal{S}_\text{AC}$ under the constraint $\mathcal{S}_\text{AB}\geq \alpha$, for some pre-selected value of $\alpha$.
	\item Evaluate the SDP over assemblages $\{\tau_{a|x}\}$ that optimises $\mathcal{S}_\text{AC}$ under the constraint $\mathcal{S}_\text{AB}\geq \alpha$.
	\item Evaluate the SDP over the measurement  $\{C_z\}$ that optimises $\mathcal{S}_\text{AC}$.
	\item Repeat points 2-4 until convergence is achieved.
\end{enumerate}
The above procedure depends on the random starting point. Therefore, we have for every choice of $\alpha$ repeated the above process for 25 different random starting points and then selected the best result. This was done for several values of $\alpha$, allowing us to probe the trade-off. The results illustrated in Fig. 4 of the main text correspond to the data table ~\ref{tabdat}.

\begin{table}[]
	\begin{tabular}{c|ccclllllllllll}
		\centering	$\mathcal{S}_\text{AB}$ & 2.00   & 2.01   & 2.02   & 2.03   & 2.04   & 2.05   & 2.06   & 2.07   & 2.08   & 2.09   & 2.10   & 2.11   & 2.12   & 2.13   \\ \cline{1-15}
		\centering	$\mathcal{S}_\text{AC}$                         & 2.1535  & 2.1334 & 2.1179 & 2.1025 & 2.0885 & 2.0753 & 2.0638 & 2.0523 & 2.0405 & 2.0310 & 2.0211 & 2.0119 & 2.0033& 1.9951
	\end{tabular}\caption{Data table for sequential CHSH violations with local randomness and PIs}\label{tabdat}
\end{table}

\section{Dimension-scalable advantage over projective instruments}\label{AppHighDim}
We consider the L\"uders instrument corresponding to Kraus operators
\begin{equation}
	K_a=\sqrt{\frac{1+\gamma}{2}}\ketbra{a}+\sqrt{\frac{1-\gamma}{2(d-1)}}\left(\openone-\ketbra{a}\right),
\end{equation}
where $\gamma\in[0,1]$ is the sharpness parameter.  The Choi representatin of the instrument is $\eta_a=\left(K_a\otimes\openone\right)\phi^+\left(K_a\otimes\openone\right)^\dagger$, which becomes
\begin{equation}
	\eta_a=\frac{1+\gamma}{2d}\ketbra{aa}+\frac{\sqrt{1-\gamma^2}}{2d\sqrt{d-1}}\sum_{l\neq a}\left(\ketbra{aa}{ll}+\ketbra{ll}{aa}\right)+\frac{1-\gamma}{2d(d-1)}\sum_{k\neq a}\sum_{l\neq a}\ketbra{kk}{ll}.
\end{equation}

We consider the mixture with dephasing noise, namely
\begin{equation}
	\eta_a^v=v\eta_a+(1-v)\eta_a^\text{noise},
\end{equation}
where $v\in[0,1]$ is the visibility and $\eta_a^\text{noise}=\frac{1}{d^2}\sum_{i=0}^{d-1}\ketbra{ii}$. We will show that the critical visibility for PI-simulation is bounded as follows, 
\begin{align}\label{F4}
	v_\text{deph}(d) \leq  \frac{2(d-1)}{-2+d\left(1+\gamma+\sqrt{(1-\gamma^2)(d-1)}\right)}. 
\end{align}
Specifically, we will first construct a  PI-simulation model that achieves $v_\text{depth}(d)$ for some values of $\gamma$. Then, we will then prove that no higher visibility can be achieved for any $\gamma$.

\subsection{Simulation model}
Consider the Choi representation of a generic PI,
\begin{equation}
	\mu_a\!=\sum_{\vec{r}} q_{\chi,\vec{r}} \nu_{a,\chi,\vec{r}},
\end{equation}
where 
\begin{equation}
	\nu_{a,\chi,\vec{r}}=\left(\Lambda_{a,\chi,\vec{r}}\otimes\openone\right)\left[E_{a|\chi,\vec{r}}\otimes \openone \phi^+ E_{a|\chi,\vec{r}}\otimes \openone\right].
\end{equation}
In the simulation, we will not use the random variable $\chi$ and therefore we discard it throughout this discussion. Our simulation uses only three classes of rank-vectors. These are
\begin{align}
	& (d,0,\ldots,0) +\text{permutations}, && (d-1,1,0,\ldots,0)+\text{permutations},  &&&(1,1,\ldots,1).
\end{align}
We refer to these classes as $C_1$, $C_2$ and $C_3$ respectively. They contain $d$, $d(d-1)$ and $1$ elements respectively.

\textbf{Class 1.} For the first class, let us denote by $\vec{s}^{(j)}$ the rank-vector that is non-zero only at position $j$ where the value is $d$.  The only relevant projective measurement is $E_{a|\vec{s}^{(j)}}=0$ if $a\neq j$ and $E_{a|\vec{s}^{(a)}}=\openone$.  When $a=j$, we select the subsequent CPTP map as the identity channel. Hence, the Choi operators become
\begin{equation}
	\nu_{a,\vec{s}^{(j)}}=\begin{cases}
		0 & \text{if } a\neq j\\
		\phi^+ &\text{if }a=j 
	\end{cases}.
\end{equation}
We select the prior as uniform, namely $q_{\vec{s}^{(j)}}=\alpha$, for every choice of $j$.

\textbf{Class 2.} Consider the second class of rank-vectors. We write $\vec{t}^{(j,k)}$ for the vector that is zero everywhere except at positions $j$ and $k$ where the values are $d-1$ and $1$ respectively. We select the projective measurements as $E_{j|\vec{t}^{(j,k)}}=\openone-\ketbra{k}$, $E_{k|\vec{t}^{(j,k)}}=\ketbra{k}$ and $E_{a|\vec{t}^{(j,k)}}=0$ otherwise. The CPTP map is always the identity channel. The Choi operators become
\begin{equation}
	\nu_{a,\vec{t}^{(j,k)}}=\begin{cases}
		\phi^+_{k} & \text{if } a=j\\
		\frac{1}{d}\ketbra{kk} &\text{if } a=k\\
		0 & \text{otherwise} 
	\end{cases},
\end{equation}
where
\begin{equation}
	\phi_j^+=\frac{1}{d}\sum_{i\neq j}\sum_{k\neq j}\ketbra{ii}{kk}
\end{equation}
is the sub-normalised maximally entangled state the $(d-1)$-dimensional subspace orthogonal to $\ket{j}$. We select a uniform prior, namely $q_{\vec{t}^{(j,k)}}=\delta$.

\textbf{Class 3.} Lastly, consider the  third class, which has only a single rank-vector, $\vec{u}=(1,\ldots,1)$. Select the projective measurement as $E_{a|\vec{u}}=\ketbra{aa}$ and the CPTP maps as the identity channel. The Choi operators become
\begin{equation}
	\nu_{a,\vec{u}}=\frac{1}{d}\ketbra{aa}.
\end{equation}
We select the prior as $q_{\vec{u}}=\beta$.

Putting it together, the simulated Choi operators become
\begin{align}
	& \mu_a=\alpha\phi^+ + \frac{\beta}{d}\ketbra{aa}+\delta\sum_{k\neq a}\phi^+_k+\delta(d-1)\frac{1}{d}\ketbra{aa},
\end{align}
where normalisation corresponds to the constraint
\begin{equation}
	\alpha|C_1|+\delta|C_2|+\beta|C_3|=d\alpha+d(d-1)\delta+\beta=1.
\end{equation}
We now need to solve the simulation equation, $\mu_a=\eta_a^v$. Upon inspection, this can be broken down into the following equations.
\begin{align}
	&v\frac{1+\gamma}{2d}+\frac{1-v}{d^2}=\frac{\alpha}{d}+\frac{\beta}{d}+\delta\frac{2(d-1)}{d}\\
	& v\frac{1-\gamma}{2d(d-1)}+\frac{1-v}{d^2}=\frac{\alpha}{d}+\delta \frac{d-2}{d}\\
	&v\frac{1-\gamma}{2d(d-1)}=\frac{\alpha}{d}+\delta\frac{d-3}{d}\\
	&v\frac{\sqrt{1-\gamma^2}}{2d\sqrt{d-1}}=\frac{\alpha}{d}+\delta\frac{d-2}{d}
\end{align}
The solution is 
\begin{align}
	& \alpha=\frac{-3 \sqrt{1-\gamma ^2}-\sqrt{1-\gamma ^2} d^2+d \left(4 \sqrt{1-\gamma ^2}+\gamma  \left(-\sqrt{d-1}\right)+\sqrt{d-1}\right)+2 \gamma  \sqrt{d-1}-2
		\sqrt{d-1}}{\sqrt{1-\gamma ^2} d^2+d \left(-\sqrt{1-\gamma ^2}+\gamma  \sqrt{d-1}+\sqrt{d-1}\right)-2 \sqrt{d-1}}\\
	&\beta=\frac{(d-1) \left(2 \sqrt{d-1}-\sqrt{1-\gamma ^2} d\right)}{\sqrt{1-\gamma ^2} d^2+d \left(-\sqrt{1-\gamma ^2}+\gamma  \sqrt{d-1}+\sqrt{d-1}\right)-2 \sqrt{d-1}}\\
	& \delta=\frac{-\sqrt{1-\gamma ^2}+\sqrt{1-\gamma ^2} d+\gamma  \sqrt{d-1}-\sqrt{d-1}}{\sqrt{1-\gamma ^2} d^2+d \left(-\sqrt{1-\gamma ^2}+\gamma 
		\sqrt{d-1}+\sqrt{d-1}\right)-2 \sqrt{d-1}},
\end{align}
and it gives the simulation visibility 
\begin{equation}
	v=\frac{2 (d-1)^{3/2}}{\sqrt{1-\gamma ^2} d^2+d \left(-\sqrt{1-\gamma ^2}+\gamma  \sqrt{d-1}+\sqrt{d-1}\right)-2 \sqrt{d-1}}.
\end{equation}
Note however that the coefficient $\beta$ is not always non-negative. Solving $\beta=0$ shows that it is a valid solution only when
\begin{equation}
	\gamma\geq \frac{d-2}{d}.
\end{equation}
Thus, the above model is valid in this regime.

\subsection{Upper bound on critical visibility}
We now prove that  no projective simulation with a visibility higher than \eqref{F4} is possible. One way to obtain an upper bound on the critical visibility is to consider the conditions formulated in Eq.~\eqref{choiA} (based on Theorem 1) for every $\gamma$ and $d$. As discussed before, we can relax this optimisation problem by substituting the Schmidt number constraint in Eq.~\eqref{choiA} with only a necessary condition on for a state having a given Schmidt number. Around Eq.~\eqref{prog}, we have suggested to use the generalised reduction map for this purpose; $\Theta_s(X)=\tr(X)\openone-\frac{1}{s}X$. The primal SDP, whose solution upper bounds, the critical visibility becomes
\begin{align}\label{primal}\nonumber
	\max_{v,\sigma} & \quad v\\
	& v\eta_a+(1-v)\eta_a^\text{noise}= \sum_{\vec{r}} \sigma_{a|\vec{r}}, \quad \forall a\\
	&\tr(\sigma_{a|\vec{r}})=q_{\vec{r}}\frac{r_a}{d},  \quad \forall a,\vec{r}\\
	& \sum_{a}\tr_{A'}(\sigma_{a|\vec{r}})=q_{\vec{r}}\frac{\openone}{d}, \quad \forall \vec{r},\\
	& (\Theta_{r_a}\otimes \openone)[\sigma_{a|\vec{r}}]\succeq 0, \quad \forall a,\vec{r}\\
	&\sigma_{a|\vec{r}}\succeq 0, \quad \forall a,\vec{r}.
\end{align}
While this can be computed for given $(\gamma,d)$, it is less straightforward to compute it as a closed expression for any $(\gamma,d)$. Therefore, we instead use the  duality theorem of SDPs to derive an analytical upper bound on the solution of the primal. This is achieved by constructing a feasible point of its dual formulation.

We use standard procedure to derive the dual SDP. The dual becomes
\begin{align}\label{dual}\nonumber
	\min_{t,B,W,Z} \quad& 1+\sum_a\tr\left(W_a\eta_a\right)\\\nonumber
	&\text{s.t. } \quad  1+\sum_a\tr\left(W_a\eta_a\right)=\sum_a\tr\left(W_a\eta^\text{noise}_a\right),\\\nonumber
	& W_a+\frac{1}{d}\tr(B_{\vec{r}})\openone+\frac{1}{r_a}Z_{a|\vec{r}}+\frac{1}{d}\left(\sum_{l=1}^N r_lt_{l|\vec{r}}\right)\openone-\openone\otimes B_{\vec{r}}-t_{a|\vec{r}}\openone-\openone\otimes Z^A_{a|\vec{r}}\succeq 0 \quad \forall a,\vec{r},\\
	& Z_{a|\vec{r}}\succeq 0 \quad \forall a,\vec{r}.
\end{align} 
Here, $Z$ and $W$ are operators on $\mathcal{H}_{A'}\otimes \mathcal{H}_A$, $B$ are operators on $\mathcal{H}_A$ and $t$ are scalars. 

Now, we construct the feasible point in the space of $(t,B,W,Z)$. To this end, we first select  $t_{a|\vec{r}}=t$ $\forall a,\vec{r}$. In fact, this eliminates all dependence on $t$ since $\sum_{l} r_l=d$. Next, we select  $Z_{a|\vec{r}}=s\phi^+$ $\forall a,\vec{r}$, for some $s\geq 0$. This ensures that the final constraint in \eqref{dual} is satisfied. Also, let us choose the operators $B_{\vec{r}}$ such that $\tr(B_{\vec{r}})=0$. The constraints in the dual have now simplified to 
\begin{align}\nonumber
	&1+\sum_a\tr\left(W_a\eta_a\right)=\sum_a\tr\left(W_a\eta^\text{noise}_a\right),\\\nonumber
	& W_a+\frac{s}{r_a}\phi^+-\openone\otimes
	B_{\vec{r}}-\frac{s}{d}\openone\succeq 0 \quad \forall a,\vec{r},\\\nonumber
	& \tr(B_{\vec{r}})=0 \quad \forall \vec{r},\\
	& s\geq 0. \label{eq:Wconst}
\end{align}
Let us now select the trace-less operators $B_{\vec{r}}$ to be diagonal: $B_{\vec{r}}=\sum_{l=1}^d c_{l,\vec{r}}\ketbra{l}$, for some real coefficients $c_{l,\vec{r}}$. Define $c_{l,\vec{r}}=\alpha (r_l-1)$ and note that $\sum_l c_{l,\vec{r}}=0$ since $\sum_l r_l=d$. Thus, we have simplified the second constraint above to
\begin{align}\label{step}
	& W_a+\frac{s}{r_a}\phi^+-\left(\frac{s}{d}-\alpha\right)\openone-\alpha\openone\otimes
	\left(\sum_{l=1}^d r_l\ketbra{l}\right)\succeq 0 \quad \forall a,\vec{r}.
\end{align}
We select the following form of $W_a$,
\begin{align}
	W_a = \beta\sum_{(i,j)\neq (a,a)}\ketbra{ij}{ij} + \alpha\sum_{j\neq a} (\ketbra{jj}{aa} + \ketbra{aa}{jj}),
\end{align}
where $\alpha$ and $\beta$ are free real variables. Using the expressions for $W_a$, $\eta_a$ and $\eta_a^{\text{noise}}$, we have 
\begin{align}
	\tr(W_a\eta^\text{noise}_a) = \beta\frac{d-1}{d} \quad \text{and} \quad \tr(W_a\eta_a) =  \frac{1}{2} \beta  (1-\gamma )+\alpha  \sqrt{\left(1-\gamma ^2\right) (d-1)}
\end{align}Fixing $\beta = -2\alpha$, then using the first constraint in \eqref{eq:Wconst}, we have
\begin{align}
	\alpha = -\frac{d}{(d-2)+d\gamma + d\sqrt{(1-\gamma ^2)(d-1)}}. %\quad \beta  = \frac{2d}{(d-2)+d\gamma + d\sqrt{(1-\gamma ^2)(d-1)}}
\end{align}
Evaluating the objective function in \eqref{dual}, we obtain precisely the expression \eqref{F4}.

However, in order to complete the proof we still need to show that the second constraint in Eq.~\eqref{eq:Wconst} can be satisfied. To this end, we select  $s=s'$ with 
\begin{align}
	s' = \frac{d^2}{(d-2) +d\gamma+d\sqrt{(d-1)(1-\gamma^2)}},
\end{align}
and for simplicity we define
\begin{align}
	O_{a,\vec r} =   W_a+\frac{s'}{r_a}\phi^+-\left(\frac{s'}{d}-\alpha\right)\openone-\alpha\openone\otimes
	\left(\sum_{l=1}^d r_l\ketbra{l}\right).
\end{align}
In the next subsection, we show that $O_{a,\vec{r}}$ is positive semidefinite for every $(a,\vec{r})$ and every $(\gamma,d)$.

\subsection{Positive-semidefiniteness of $O_{a,\vec r}$}
We show that all eigenvalues of $O_{a,\vec r}$ are non-negative. To this end, we consider its characteristic polynomial. Upon inspection, it can be expressed as 
\begin{align}
	\det\left(O_{a,\vec r} - \lambda\mathds 1 \right) = 	\mathds P_{d,a}^{\vec r}\prod_{j=1}^{d}( \lambda - \frac{s'}{d}r_j)^{d-1}
\end{align}where $\mathds P_{d,a}^{\vec r} = (A_0+A_1\tilde \lambda +A_2\tilde \lambda^2+...+A_d\tilde \lambda^d)$ is a polynomial of degree at most $d$, $r_j$ are elements of the chosen rank vector $\vec r$ and $\tilde \lambda=\frac{d\,r_a}{s'} \lambda$. To express the coefficients of $\mathds P_{d,a}^{\vec r}$, we define $S_{a}^p$ as the sum of all unique $p$-products of the elements of the rank vector excluding $r_a$, such that in each product, an individual element appears at most once. For example, in $d=4$, $S_1^2 = r_2 r_3+ r_2r_4+r_3r_4$ and in $d=5$, $S_2^3 = r_1r_3r_4+r_1r_3r_5+r_3r_4r_5$. Also, by definition, $S^0_a=1$ and $S^p_a=0$ $\forall\, p<0$ or $p\geq d$. With this definition, the coefficients can be expressed as

\begin{equation}
	A_n =(-1)^{d+n} \left[r_a^{d-n-1}(r_a^2-2r_a+(n+1))S_a^{d-n-1} + r_a^{d-n}S^{d-n}_a\right]  ,\,\,n=0,1,...,d
\end{equation}

To confirm that this is indeed the form of the coefficients, we have numerically verified it till $d=12$ for all possible rank vectors. 

The positive-semidefinite property of $O_{a,\vec r}$ can be seen by inspecting its eigenvalues. The  eigenvalues $s'\,r_j/d$ are evidently all non-negative.  We require the rest of the eigenvalues, which are the roots of the polynomial $\mathds P_{d,a}^{\vec r}$, to also be non-negative. First note that $A_0=(-1)^dr_a^{d-2}(r_a-1)^2\Pi_{j=1}^dr_j=0$ $\forall\, \vec r$. $\Pi_{j=1}^dr_j$ is nonzero only for one rank vector, $(1,1,...,1)$. However, for this rank vector, $r_a-1=0$ $\forall\,a$, implying $A_0=0$. This implies that zero is a root of $\mathds P_{d,a}^{\vec r}$. For generality, we assume that all other coefficients of $\mathds P_{d,a}^{\vec r}$ are non-zero.  For many classes of rank vectors, other coeffecients can be zero. Such situations can be handled in a completely analogous manner, albeit case-by-case. First note that $O_{a,\vec r}$ is Hermitian and therefore the roots of $\mathds P_{d,a}^{\vec r}$ must be real. Then, by Descartes' rule of signs \cite{Descartes1637,Sullivan2017}, the number of positive roots of $\mathds P_{d,a}^{\vec r}$ is exactly equal to the number of sign changes in the sequence of its coefficients. Note the following property in the structure of the coefficients: $r_a^{d-n-1}(r_a^2-2r_a+(n+1))S_a^{d-n-1} + r_a^{d-n}S^{d-n}_a\geq0$. The left term's non-negativity follows from $(r_a^2-2r_a+(n+1))\geq 0$ $\forall\, r_a\geq 0$ and $n\geq 0$. The right term is naturally non-negative.  This means that the sign of $A_n$ is determined solely by $(-1)^{d+n}$. Therefore, we have exactly $d-1$ sign changes in the non-zero coefficients of $\mathds P_{d,a}^{\vec r}$. This implies that all remaining $d-1$ of its roots are positive.  \qed

\section{Projective simulation for worst-case noise}\label{AppWorst}
We consider the L\"uders instrument corresponding to Kraus operators
\begin{equation}
	K_a=\sqrt{\frac{1+\gamma}{2}}\ketbra{a}+\sqrt{\frac{1-\gamma}{2(d-1)}}\left(\openone-\ketbra{a}\right),
\end{equation}
where $\gamma\in[0,1]$ is the sharpness parameter.  The Choi representation of the instrument is $\eta_a=\left(K_a\otimes\openone\right)\phi^+\left(K_a\otimes\openone\right)^\dagger$, which becomes
\begin{equation}
	\eta_a=\frac{1+\gamma}{2d}\ketbra{aa}+\frac{\sqrt{1-\gamma^2}}{2d\sqrt{d-1}}\sum_{l\neq a}\left(\ketbra{aa}{ll}+\ketbra{ll}{aa}\right)+\frac{1-\gamma}{2d(d-1)}\sum_{k\neq a}\sum_{l\neq a}\ketbra{kk}{ll}.
\end{equation}
We consider the mixture with worst-case noise, 
\begin{equation}
	\eta_a^v=v\eta_a+(1-v)\eta_a^\text{worst},
\end{equation}
where $v\in[0,1]$ is the visibility and $\eta_a^\text{worst}$ is some arbitrary noise, whose form we will determine. We will show that there exists a PI-simulation that achieves the visibility  
\begin{align}\label{visi}
	v=  \frac{d-1}{-\gamma -\sqrt{\left(1-\gamma ^2\right) (d-1)}+\sqrt{d(1-\gamma ) \left((\gamma +1) d-2 \left(\gamma +\sqrt{\left(1-\gamma ^2\right) (d-1)}\right)\right)}+d}.
\end{align}
Thus, this constitutes a lower bound on the the critical visibility, $v\leq v_\text{worst}(d)$, but we conjecture it to be optimal for any $d$ and $\gamma$. This conjecture is supported by evaluating the SDP for PIs for $d=2,3$ and many selected values of $\gamma$.

It can be checked that the minimum visibility for given $d$ lies at $\gamma=1/\sqrt{d}$. The visibility then becomes 
\begin{align}
	v^{\text{min}}_{\text{worst}}(d) = \frac{1}{2}\left(1+\frac{1}{\sqrt{d}} \right).
\end{align}
Thus, for any $d$ and any $\gamma$, there always exists a PI-simulation that achieves visibility $v=\frac{1}{2}$.

We now show how to construct the PI model. Consider the Choi representation of a generic PI,
\begin{equation}
	\mu_a\!=\sum_{\vec{r}} q_{\chi,\vec{r}} \nu_{a,\chi,\vec{r}},
\end{equation}
where 
\begin{equation}
	\nu_{a,\chi,\vec{r}}=\left(\Lambda_{a,\chi,\vec{r}}\otimes\openone\right)\left[E_{a|\chi,\vec{r}}\otimes \openone \phi^+ E_{a|\chi,\vec{r}}\otimes \openone\right].
\end{equation}
In the simulation, we will not use the random variable $\chi$ and therefore we discard it throughout this discussion. Our simulation uses only two classes of rank-vectors. These are
\begin{align}
	(d,0,...,0) +\text{permutations} \quad \text{and} \quad (1,1,...,1). 
\end{align}For the first set of rank-vectors, the associated Choi operators corresponding to rank $d$ are identical and given by $\alpha \phi^+$. For the other rank-vector, the associated Choi operators are given by $\beta\ketbra{aa}{aa}$, corresponding to outcome $a$. Then, the simulated Choi operators become $\mu_a = \alpha \phi^+ + \beta \ketbra{aa}{aa}$. Furthermore, we take the following noise model,
\begin{align}
	\eta_a^{\text{worst}} = x_1\ketbra{aa}{aa} -\sqrt{x_1 x_2} \sum_{j\neq a} (\ketbra{aa}{jj}+\ketbra{jj}{aa})  +x_2\sum_{k\neq a} \sum_{l\neq a}\ketbra{kk}{ll} 
\end{align}We now need to solve the simulation equation $\mu_a=\eta_a^v$. Upon inspection, this can be broken down into the following equations,
\begin{equation}
	\begin{cases}
		&v	\frac{(\gamma +1) }{2 d}+(1-v)x_1=\frac{\alpha }{d} +\beta\\
		&v\frac{\sqrt{1-\gamma ^2} }{2 d \sqrt{d-1}}-(1-v)\sqrt{x_1 x_2}=\frac{\alpha }{d}\\
		&v\frac{(1-\gamma ) }{2 d (d-1)}+ (1-v)x_2=\frac{\alpha }{d} \\
		&(d-1) x_2+x_1=\frac{1}{d}\\
		&\alpha + \beta = \frac{1}{d},
	\end{cases}
\end{equation}where the last two equations are normalisation constraints on the noise and the simulation Choi operators, respectively. The above set of equations can be solved to give the visibility in Eq.~\eqref{visi}. The other parameters are given by
\begin{equation}
	\begin{aligned}
		&\alpha = \frac{v\sqrt{\left(1-\gamma ^2\right) (d-1)} +(d-1)(1-v\gamma)}{2 d(d-1) }\\
		&\beta = \frac{ -v\sqrt{\left(1-\gamma ^2\right) (d-1)}+(d-1)(1+v\gamma)}{2 d(d-1)}\\
		& x_1 = \frac{-v \sqrt{\left(1-\gamma ^2\right) (d-1)} - v\gamma + d (1-v) + 1}{2 d^2 (1-v)}\\
		& x_2 = \frac{v\sqrt{\left(1-\gamma ^2\right) (d-1)}+v\gamma + d (1-v) - 1}{2 (d-1) d^2 (1-v)}. \\
	\end{aligned}
\end{equation}

\end{document}